\documentclass[11pt]{article}
\usepackage{amssymb}
\usepackage{amsfonts}
\usepackage{indentfirst}
\usepackage{amsopn}
\usepackage{amsmath}
\usepackage{amsthm}
\usepackage{amscd}
\usepackage{graphicx}

\makeatletter\@addtoreset {equation}{section}\makeatother

\newtheorem{theorem}{Theorem}

\newtheorem{proposition}{Proposition}

\begin{document}

\title{\bf Rogue periodic waves of the mKdV equation}

\author{Jinbing Chen$^{1,2}$ and Dmitry E. Pelinovsky$^{2,3}$ \\
{\small \it $^1$ School of Mathematics, Southeast University, Nanjing, Jiangsu 210096, P.R. China} \\
{\small \it $^2$ Department of Mathematics, McMaster University, Hamilton, Ontario, Canada, L8S 4K1 } \\
{\small \it $^3$ Department of Applied Mathematics,
Nizhny Novgorod State Technical University} \\
{\small \it 24 Minin street, 603950 Nizhny Novgorod, Russia} }

\date{\today}
\maketitle

\begin{abstract}
Rogue periodic waves stand for rogue waves on the periodic background.
Two families of traveling periodic waves of the modified Korteweg--de Vries (mKdV) equation
in the focusing case are expressed by the Jacobian elliptic functions {\em dn} and {\em cn}.
By using one-fold and two-fold Darboux transformations, we construct explicitly
the rogue periodic waves of the mKdV equation. Since the {\em dn}-periodic wave
is modulationally stable with respect to long-wave perturbations, the
``rogue" {\em dn}-periodic solution is not a proper rogue wave on the periodic background
but rather a nonlinear superposition of an algebraically decaying soliton and
the {\em dn}-periodic wave. On the other hand, since the {\em cn}-periodic wave
is modulationally unstable with respect to long-wave perturbations, the
rogue {\em cn}-periodic wave is a proper rogue wave on the periodic background,
which generalizes the classical rogue wave (the so-called Peregrine's breather) of
the nonlinear Schr\"{o}dinger (NLS) equation. We compute the magnification factor
for the rogue {\em cn}-periodic wave of the mKdV equation and show
that it remains  for all amplitudes the same as in the small-amplitude NLS limit.
As a by-product of our work, we find explicit expressions for the periodic eigenfunctions
of the AKNS spectral problem associated with the {\em dn} and {\em cn} periodic waves of the mKdV equation.
\end{abstract}

\section{Introduction}

Simplest models for nonlinear waves in fluids such as the nonlinear Schr\"{o}dinger equation (NLS),
the Korteweg--de Vries equation (KdV), and the modified Korteweg--de Vries equation (mKdV) have many things in common.
First, they appear to be integrable by using the inverse scattering transform method for the same
AKNS (Ablowitz--Kaup--Newell--Segur) spectral problem \cite{AKNS}. Second, there exist asymptotic transformations
of one nonlinear evolution equation to another nonlinear evolution equation, e.g. from defocusing NLS to KdV
and from KdV and focusing mKdV to the defocusing and focusing NLS respectively \cite{KZ}.

Modulation instability of the constant-wave background in the focusing NLS equation has been a paramount concept
in the modern nonlinear physics \cite{OZ}. More recently, spectral instability of the periodic waves expressed by
the elliptic functions {\em dn} and {\em cn} has been investigated in the focusing NLS \cite{DS}
(see also \cite{LeCoz,Lafortune}). Regarding periodic waves in the focusing mKdV equation,
it was found that the {\em dn}-periodic waves are modulationally stable with respect to the long-wave
perturbations, whereas the {\em cn}-periodic waves are modulationally unstable \cite{BJK,BHJ} (see also \cite{DN}).

The outcome of the modulation instability in the focusing NLS equation
is the emergence of the localized  spatially-temporal patterns on the background of the unstable periodic
or quasi-periodic waves (see review in \cite{Calini}). Such spatially-temporal patterns are known under
the generic name of {\em rogue waves} \cite{Taki}.

In the simplest setting of the constant wave background, the rogue waves are expressed as rational solutions of the
NLS equation. Explicit expressions for such rational solutions have been obtained by using
available algebraic constructions such as applications of the multi-fold Darboux transformations \cite{Akh,Matveev,JYang}.
For example, if the focusing NLS equation is set in the form
\begin{equation}
i \psi_t + \psi_{xx} + 2 (|\psi|^2 - 1) \psi = 0,
\label{nls}
\end{equation}
then the {\em classical rogue wave} up to the translations in $(x,t)$ is given by
\begin{equation}
\label{rogue-basic}
\psi(x,t) = 1 - \frac{4 (1+4it)}{1 + 4 x^2 + 16 t^2}.
\end{equation}
As $|t| + |x| \to \infty$, the rogue wave (\ref{rogue-basic}) approaches the constant wave
background $\psi_0(x,t) = 1$. On the other hand, at $(x,t) = (0,0)$, the rogue wave reaches the
maximum at $|\psi(0,0)| = 3$, from which we define the {\em magnification factor} of the
constant wave background to be $M_0 = 3$. The rogue wave (\ref{rogue-basic}) was derived by
Peregrine \cite{Peregrine} as an outcome of the modulational instability of the constant-wave
background and is sometimes referred to as {\em Peregrine's breather}.

Rogue waves over nonconstant backgrounds (e.g., the periodic waves or the two-phase solutions)
were addressed only recently in the context of the focusing NLS equation (\ref{nls}).
Computations of such rogue waves rely on the numerical implementation of
the B\"{a}cklund transformation to the periodic waves \cite{Kedziora} or
the two-phase solutions \cite{CalSch} of the NLS. Further analytical work to characterize
the general two-phase solutions of the NLS can be found in \cite{Wright} and in \cite{Tovbis1,Tovbis2}.

The purpose of this work is to obtain exact solutions for the rogue waves on the periodic background,
which we name here as {\em rogue periodic waves}. Computations of such rogue waves are developed
by an analytical algorithm with precise characterization of the periodic and non-periodic
eigenfunctions of the AKNS spectral problem at the periodic wave. Although our computations
are reported in the context of the focusing mKdV equation, the algorithm can be applied
to other nonlinear evolution equations associated to the AKNS spectral problem such as the NLS.

Hence we consider the focusing mKdV equation written in the normalized form
\begin{equation}\label{mKdV}
u_t+6u^2u_x+u_{xxx}=0.
\end{equation}
Some particular rational and trigonometric solutions of the mKdV were recently constructed in
\cite{Choudhury} and discussed in connection to rogue waves of the NLS. In comparison with \cite{Choudhury},
the novelty of our work is to obtain the rogue periodic waves expressed by
the Jacobian elliptic functions and to investigate how these rogue periodic waves generalize
in the small-amplitude limit the classical rogue wave (\ref{rogue-basic}).
In particular, we shall compute explicitly the magnification factor for the rogue periodic waves
that depends on elliptic modulus of the Jacobian elliptic functions.

There are two particular periodic wave solutions of the mKdV. One solution is strictly positive
and is given by the {\em dn} elliptic function. The other solution is sign-indefinite and
is given by the {\em cn} elliptic function. Up to the translations in $(x,t)$ as well as
a scaling transformation, the positive solution is given by
\begin{equation}
u_{\rm dn}(x,t) = {\rm dn}(x-ct;k), \quad c = c_{\rm dn}(k) := 2 - k^2,
\label{elliptic-dn}
\end{equation}
whereas the sign-indefinite solution is given by
\begin{equation}
u_{\rm cn}(x,t) = k {\rm cn}(x-ct;k), \quad c = c_{\rm cn}(k) := 2k^2 - 1.
\label{elliptic-cn}
\end{equation}
In both cases, $k \in (0,1)$ is elliptic modulus, which defines
two different asymptotic limits.
In the limit $k \to 0$, we obtain
\begin{equation}
\label{elliptic-dn-small}
u_{\rm dn}(x,t) \sim 1 - \frac{1}{2} k^2\sin^2(x - 2t)
\end{equation}
and
\begin{equation}
\label{elliptic-cn-small}
u_{\rm cn}(x,t) \sim k \cos(x + t)
\end{equation}
which are understood in the sense of the Stokes expansion of the periodic waves.
As is well-known \cite{PelPel,Parkes}, the mKdV equation can be reduced asymptotically to the NLS equation
in the small-amplitude limit. The {\em cn}-periodic wave of the mKdV in the limit $k \to 0$
is reduced to the constant wave background $\psi_0$ of the NLS equation (\ref{nls}),
which is modulationally unstable with respect to the long-wave perturbations.
Hence, the {\em cn}-periodic wave for the mKdV generalizes the constant wave background of the NLS
and inherits modulational instability with respect to long-wave perturbations.
The {\em dn}-periodic wave has nonzero mean value, which is large enough to make the {\em dn}-periodic
wave modulationally stable with respect to long-wave perturbations \cite{PelPel}.

In the limit $k \to 1$, both Jacobian elliptic functions (\ref{elliptic-dn}) and (\ref{elliptic-cn})
converges to the normalized mKdV soliton
\begin{equation}
\label{elliptic-dn-cn-large}
u_{\rm dn}(x,t), u_{\rm cn}(x,t) \to u_{\rm soliton}(x,t) = {\rm sech}(x-t).
\end{equation}
Very recently, the rogue waves of the mKdV built from a superposition of slowly interacting nearly identical
solitons were considered numerically \cite{PelSh} and analytically \cite{PelSl}. It was found
in these studies that the magnification factor of the rogue waves built from $N$ nearly identical solitons is
exactly $N$.

In our work, we compute the rogue periodic waves for the {\em dn} and {\em cn} Jacobian elliptic functions
with the following algorithm. First, by using the algebraic technique based on the nonlinearization of the Lax pair \cite{Cao1},
we obtain the explicit expressions for the eigenvalues $\lambda$ with ${\rm Re}(\lambda) > 0$
and the associated periodic eigenfunctions in the AKNS spectral problem associated with the Jacobian elliptic functions.
These eigenvalues correspond to the branch points of the continuous bands, when the AKNS spectral problem with the periodic potentials is
considered on the real line with the help of the Floquet--Bloch transform \cite{Calini}. For each periodic eigenfunction,
we construct the second, linearly independent solution of the AKNS spectral problem, which is not periodic but
linearly growing in $(x,t)$. The latter eigenfunction is expressed in terms of integrals of the Jacobian elliptic functions
and hence it is not explicit. Finally, by using the one-fold and two-fold Darboux transformations \cite{GuHuZhou}
with the nonperiodic eigenfunctions of the AKNS spectral problem, we obtain the rogue periodic waves.
Although the resulting solutions are not explicit, we prove that these solutions approach
the {\em dn} and {\em cn}  periodic waves as $|x|+|t| \to \infty$ almost everywhere
and that they have maximum at the origin $(x,t) = (0,0)$, where the magnification factors can be computed
in the explicit form.

Figure \ref{fig-1} shows the ``rogue" {\em dn}-periodic wave for $k = 0.5$ (left) and $k = 0.99$ (right).
We write the name of ``rogue wave" in commas, because the solution is not a proper rogue wave, the latter one
is supposed to appear from nowhere and to disappear without a trace as time evolves \cite{Taki}. Instead, we obtain a solution
that corresponds to a nonlinear superposition of the algebraically decaying soliton of the mKdV \cite{GrimshawPel} and
the {\em dn}-periodic wave, hence the maximal amplitude is brought by the algebraic soliton from infinity.
This outcome of our algorithm is related to the fact that
the {\em dn}-periodic wave in the mKdV is modulationally stable with respect to
the long-wave perturbations \cite{BHJ}. Indeed, it is argued in \cite{Wabnitz} on several examples
involving the constant wave background that the rogue wave solutions exist only in the parameter regions
where the constant wave background is modulationally unstable.

Figure \ref{fig-2} shows the rogue {\em cn}-periodic wave for $k = 0.5$ (left) and $k = 0.99$ (right).
This solution is a proper rogue wave on the periodic background because it appears from nowhere and disappears
without a trace as time evolves. The existence of such rogue periodic wave is related to the fact
that the {\em cn}-periodic wave in the mKdV is modulationally unstable with respect to
the long-wave perturbations \cite{BHJ}.

\begin{figure}[htb]
\begin{center}
\includegraphics[height=4cm]{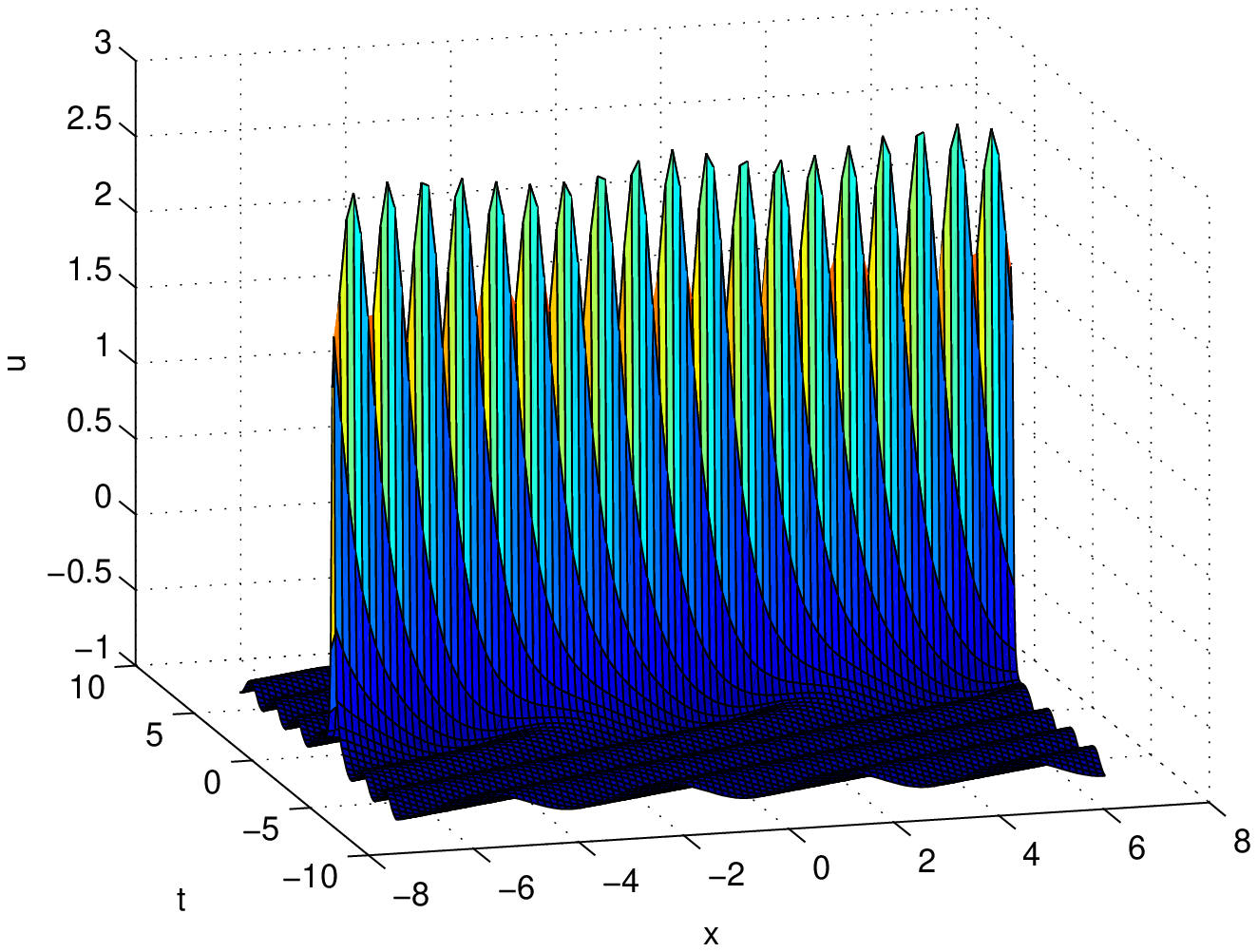}
\includegraphics[height=4cm]{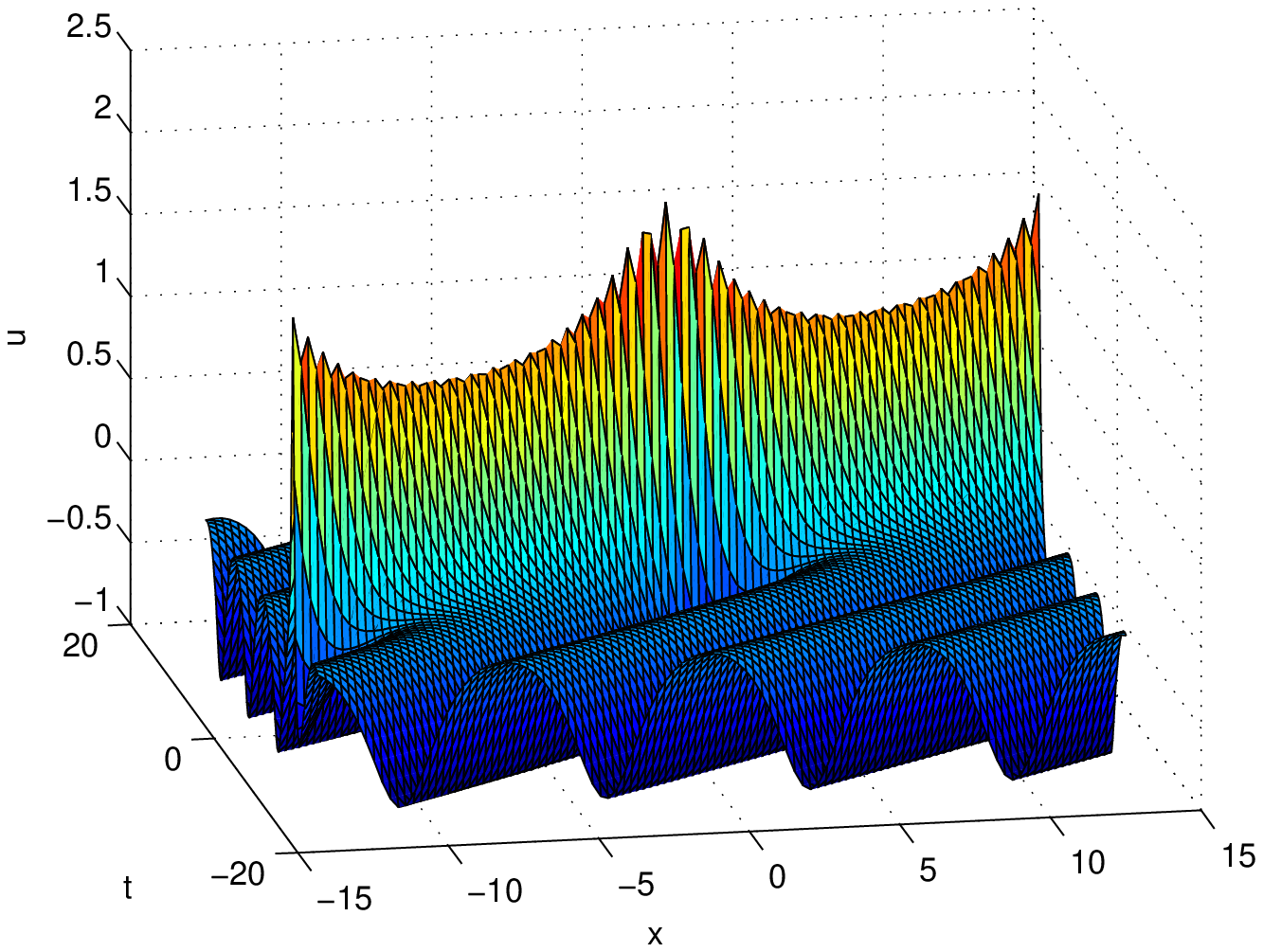}
\end{center}
\caption{The ``rogue" {\em dn}-periodic wave of the mKdV for $k = 0.5$ (left) and $k = 0.99$ (right). }
\label{fig-1}
\end{figure}

\begin{figure}[htb]
\begin{center}
\includegraphics[height=4cm]{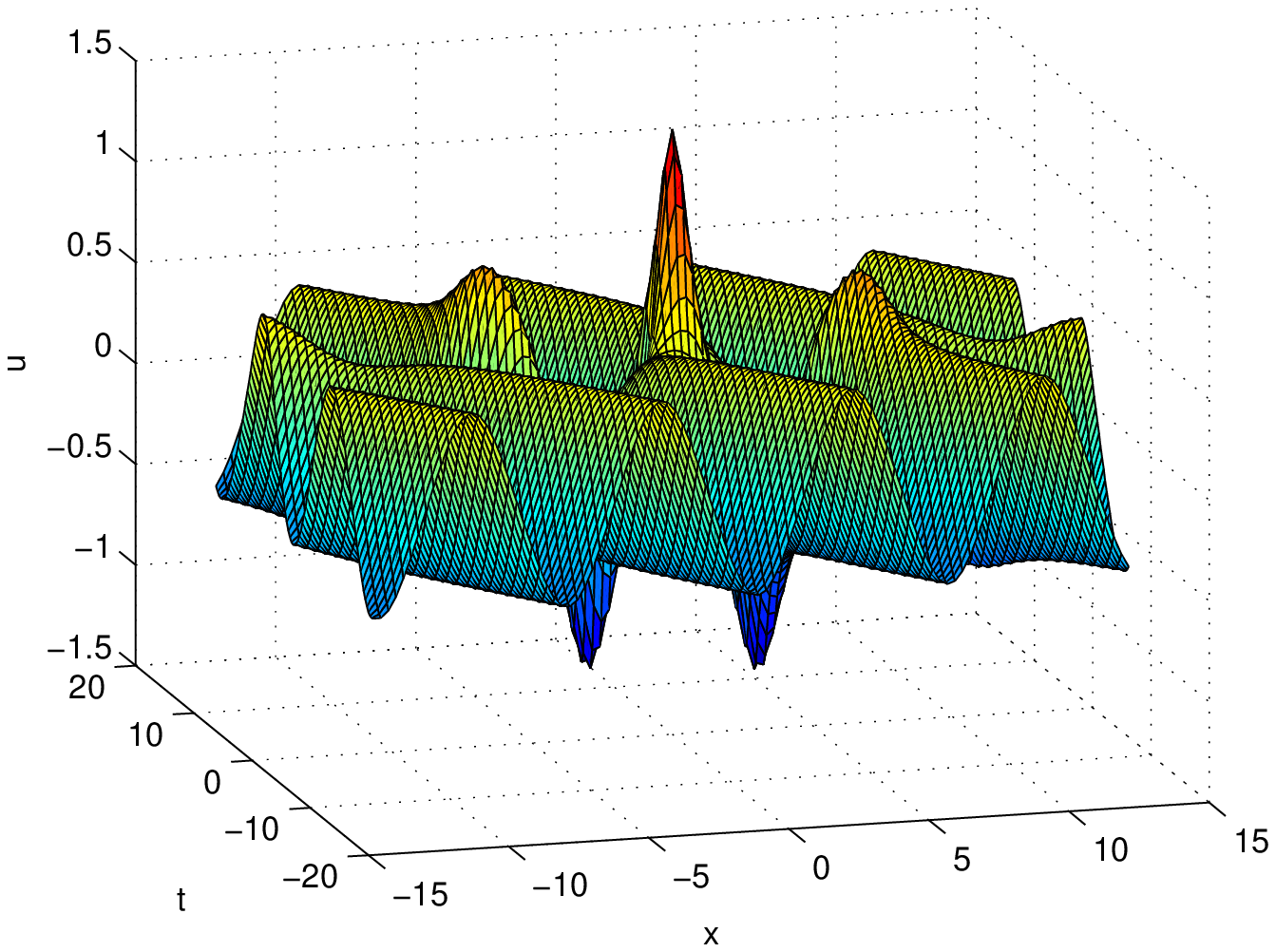}
\includegraphics[height=4cm]{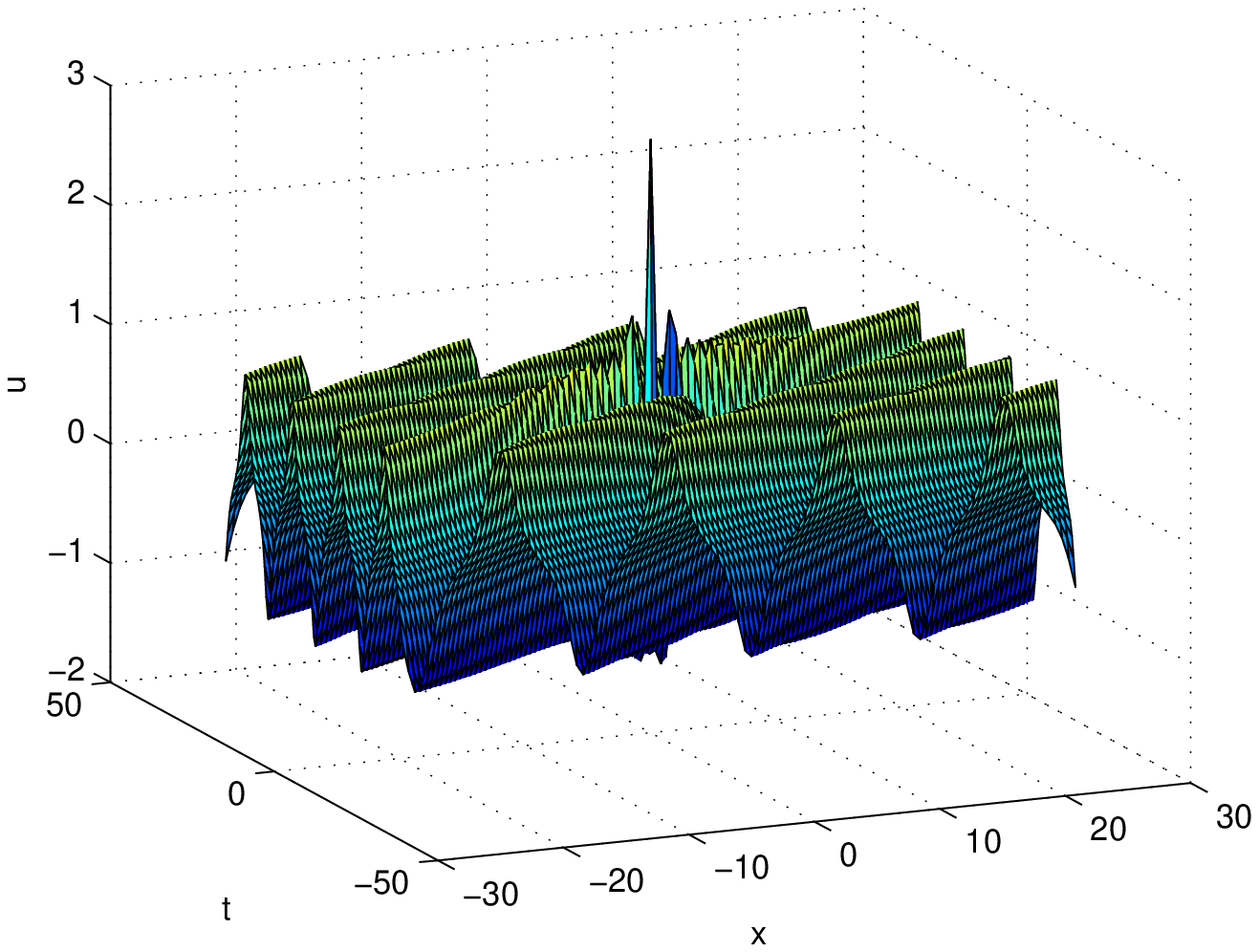}
\end{center}
\caption{The rogue {\em cn}-periodic wave of the mKdV for $k = 0.5$ (left) and $k = 0.99$ (right). }
\label{fig-2}
\end{figure}

The magnification factors of the rogue periodic waves can be computed in the explicit form:
\begin{equation}
\label{magnification-factor}
M_{\rm dn}(k) = 2 + \sqrt{1 - k^2}, \quad M_{\rm cn}(k) = 3, \quad k \in [0,1].
\end{equation}
It is remarkable that the magnification factor $M_{\rm cn}(k) = 3$ is independent of the wave amplitude
in agreement with $M_0 = 3$ for the classical rogue wave (\ref{rogue-basic})
thanks to the small-amplitude asymptotic limit (\ref{elliptic-cn-small}). At the same time
$M_{\rm dn}(k) \in [2,3]$ and $M_{\rm dn}(k) \to 3$ as $k \to 0$ due to the fact
that the limit (\ref{elliptic-dn-small}) gives the same potential to the AKNS spectral problem
as the constant wave background $\psi_0(x,t) = 1$ of the NLS equation (\ref{nls}).

In the soliton limit (\ref{elliptic-dn-cn-large}), $M_{\rm dn}(k) \to 2$ as $k \to 1$
in agreement with the recent results in \cite{PelSh,PelSl}. Indeed, the ``rogue" {\em dn}-periodic wave
degenerates as $k \to 1$ to the two-soliton solutions constructed of two nearly identical solitons. Such solutions
are constructed by the one-fold Darboux transformation from the one-soliton solutions,
when the eigenfunction of the AKNS spectral problem is nondecaying (exponentially growing) \cite{MizPel}.
Therefore, the magnification factor $M_{\rm dn}(1) = 2$ is explained by the weak interaction between
two nearly identical solitons. On the other hand, $M_{\rm cn}(1) = 3$ is explained by the fact that the
rogue {\em cn}-periodic wave is built from the two-fold Darboux transformation,
hence it degenerates as $k \to 1$ to the three-soliton solutions constructed
of three nearly identical solitons \cite{PelSl}.

The paper is organized as follows. Section 2 gives details of the periodic
eigenfunctions of the AKNS spectral problem associated with the {\em dn} and {\em cn} Jacobian elliptic functions.
The non-periodic functions of the AKNS spectral problem are computed in Section 3.
Section 4 presents the general $N$-fold Darboux transformation for the mKdV equation and the explicit
formulas for the one-fold and two-fold Darboux transformations. The rogue {\em dn}-periodic
and {\em cn}-periodic waves of the mKdV are constructed in Sections 5 and 6 respectively.
Appendix A gives a proof of the $N$-fold Darboux transformations in the explicit form.

\vspace{0.25cm}

{\bf Acknowledgement.} The authors thank E.N. Pelinovsky for suggesting the problem of computing the rogue
periodic waves in the mKdV. J. Chen is grateful to the Department of Mathematics of McMaster University
for the generous hospitality during his visit. J. Chen is supported by the National Natural Science Foundation
of China (No. 11471072) and the Jiangsu  Overseas Research $\&$ Training Programme for University Prominent Young
$\&$ Middle-aged Teachers and Presidents (No. 1160690028). D.E. Pelinovsky is supported by
the state task of Russian Federation in the sphere of scientific activity (Task No. 5.5176.2017/8.9).

\section{Periodic eigenfunctions of the AKNS spectral problem}

The mKdV equation (\ref{mKdV}) is obtained as a compatibility condition of the following Lax
pair of two linear equations for the vector $\varphi = (\varphi_1,\varphi_2)^t$:
\begin{equation}\label{3.2}
\varphi_x = U(\lambda,u) \varphi, \quad
U(\lambda,u) =\left(\begin{array}{cc} \lambda&u\\ -u&-\lambda\\ \end{array}\right),
\end{equation}
and
\begin{equation}\label{3.3}
\varphi_t = V(\lambda,u) \varphi, \quad
V(\lambda,u) = \left(\begin{array}{cc}
-4\lambda^3-2\lambda u^2&-4\lambda^2u-2\lambda u_x-2u^3-u_{xx}\\
4\lambda^2u-2\lambda u_x+2u^3+u_{xx}&4\lambda^3+2\lambda u^2\\
\end{array}\right).
\end{equation}
The first linear equation (\ref{3.2}) is referred to as the AKNS spectral problem as
it defines the spectral parameter $\lambda$ for a given potential $u(x,t)$ at a frozen time $t$, e.g. at $t = 0$.
By using the Pauli matrices
$$
\sigma_1 = \left( \begin{array}{cc} 0 & 1 \\ 1 & 0 \end{array} \right), \quad
\sigma_2 = \left( \begin{array}{cc} 0 & -i \\ i & 0 \end{array} \right), \quad
\sigma_3 = \left( \begin{array}{cc} 1 & 0 \\ 0 & -1 \end{array} \right),
$$
we can rewrite $U(\lambda,u)$ and $V(\lambda,u)$ in (\ref{3.2}) and (\ref{3.3}) in the form
\begin{eqnarray}
\label{operator-U}
U(\lambda,u) & = & \lambda \sigma_3 + u \sigma_3 \sigma_1, \\
\label{operator-V}
V(\lambda,u) & = & -(4 \lambda^3 + 2 \lambda u^2) \sigma_3 - 4 \lambda^2 u \sigma_3 \sigma_1 - 2 \lambda u_x \sigma_1 - (2 u^3 + u_{xx}) \sigma_3 \sigma_1.
\end{eqnarray}

If $u$ is either {\em dn} or {\em cn} Jacobian elliptic functions (\ref{elliptic-dn}) and (\ref{elliptic-cn}),
the potentials are $L$-periodic in $x$ with the period $L = 2K(k)$ for {\em dn}-functions and $L = 4 K(k)$ for {\em cn}-functions,
where $K(k)$ is the complete elliptic integral. If the AKNS spectral problem (\ref{3.2}) is considered in the space of
$L$-periodic functions, then the admissible set for the spectral parameter $\lambda$ is discrete as
the AKNS spectral problem has a purely point spectrum.

In the case of periodic or quasi-periodic potentials $u$, the algebraic technique based on the nonlinearization
of the Lax pair \cite{Cao1} (see also applications in \cite{Cao2,Cao4,Chen,Cao3})
can be used to obtain explicit solutions for the eigenfunctions of the AKNS spectral problem
related to the particular eigenvalues $\lambda$ with ${\rm Re}(\lambda) > 0$.
Below we simplify the general method in order to obtain particular solutions of the AKNS spectral problem for the
periodic waves in the focusing mKdV equation (\ref{mKdV}).
The following two propositions represent the explicit expressions for eigenvalues and periodic eigenfunctions of system
(\ref{3.2}) and (\ref{3.3}) related to the travelling periodic wave solution of the mKdV.

\begin{proposition}
\label{prop-AKNS}
Let $u$ be a travelling wave solution of the mKdV equation (\ref{mKdV}) satisfying
\begin{equation}
\label{ode}
\frac{d^2 u}{d x^2} + 2 u^3 = c u, \quad \left(\frac{du}{dx} \right)^2 + u^4 = c u^2 + d,
\end{equation}
where $c$ and $d$ are real constants parameterized by
\begin{equation}
\label{ode-parameters}
c = 4 \lambda_1^2 + 2 E_0, \quad d = - E_0^2
\end{equation}
with possibly complex $\lambda_1$ and $E_0$. Then, there exists a solution $\varphi = (\varphi_1,\varphi_2)^t$
of the AKNS spectral problem (\ref{3.2}) with $\lambda = \lambda_1$ such that
\begin{equation}
\label{AKNS-solution}
\varphi_1^2 + \varphi_2^2 = u, \quad \varphi_1^2 - \varphi_2^2 = \frac{1}{2\lambda_1} \frac{du}{dx},
\quad 4 \lambda_1 \varphi_1 \varphi_2 = E_0 - u^2.
\end{equation}
In particular, if $u$ is periodic in $x$, then $\varphi$ is periodic in $x$.
\end{proposition}

\begin{proof}
Following \cite{Cao1}, we set $u = \varphi_1^2 + \varphi_2^2$ and consider a nonlinearization
of the AKNS spectral problem (\ref{3.2}) given by the Hamiltonian system
\begin{eqnarray}
\label{Ham-sys}
\frac{d \varphi_1}{dx} = \frac{\partial H}{\partial \varphi_2}, \quad
\frac{d \varphi_2}{dx} = -\frac{\partial H}{\partial \varphi_1},
\end{eqnarray}
which is related to the Hamiltonian function
\begin{equation}
\label{Ham-fun}
H(\varphi_1,\varphi_2) = \frac{1}{4} (\varphi_1^2+\varphi_2^2)^2 + \lambda_1 \varphi_1 \varphi_2 = \frac{1}{4} E_0,
\end{equation}
where $E_0$ is constant in $x$. It follows from (\ref{Ham-fun})
that $4 \lambda_1 \varphi_1 \varphi_2 = E_0 - u^2$. Also note that
$$
\frac{du}{dx} = 2 \left( \varphi_1 \frac{d \varphi_1}{d x} + \varphi_2 \frac{d \varphi_2}{dx} \right)
= 2 \lambda_1 (\varphi_1^2 - \varphi_2^2),
$$
so that all three equations in (\ref{AKNS-solution}) are satisfied by the construction.

Let us introduce
$$
Q(\lambda) = \left(\begin{array}{cc}
\lambda & \varphi_1^2 + \varphi_2^2\\
-\varphi_1^2 - \varphi_2^2 &-\lambda
\end{array}\right), \quad
W(\lambda) = \left(\begin{array}{cc}
W_{11}(\lambda) & W_{12}(\lambda)\\
W_{12}(-\lambda) & -W_{11}(-\lambda) \end{array}\right),
$$
with
\begin{eqnarray*}
W_{11}(\lambda) & = & 1 - \frac{\varphi_1 \varphi_2}{\lambda - \lambda_1} + \frac{\varphi_1 \varphi_2}{\lambda + \lambda_1}, \\
W_{12}(\lambda) & = & \frac{\varphi_1^2}{\lambda-\lambda_1} + \frac{\varphi_2^2}{\lambda + \lambda_1}.
\end{eqnarray*}
Then, one can check directly that the Lax equation
$$
\frac{d}{dx} W(\lambda) = Q(\lambda) W(\lambda) - W(\lambda) Q(\lambda),
$$
is satisfied for every $\lambda \in \mathbb{C}$ if and only if
$(\varphi_1,\varphi_2)$ satisfies (\ref{Ham-sys}). In particular, the $(1,2)$-entry in the above relations yields the equation
\begin{equation}
\label{ham-1-2}
\frac{d}{dx} W_{12}(\lambda) = 2 \lambda W_{12}(\lambda) - 2 (\varphi_1^2+\varphi_2^2) W_{11}(\lambda)
\end{equation}
and the representation
\begin{equation}
\label{ham-repr}
W_{12}(\lambda) = \frac{1}{\lambda^2 - \lambda_1^2} \left[ \lambda (\varphi_1^2 + \varphi_2^2) + \lambda_1 (\varphi_1^2 - \varphi_2^2) \right]
=: \frac{(\lambda-\mu) (\varphi_1^2 + \varphi_2^2)}{a(\lambda)},
\end{equation}
with
\begin{equation}
\label{mu-representation}
a(\lambda) := \lambda^2 - \lambda_1^2 \quad \mbox{\rm and} \quad
\mu := -\lambda_1 \frac{\varphi_1^2 - \varphi_2^2}{\varphi_1^2 + \varphi_2^2} = -\frac{1}{2u} \frac{du}{dx}.
\end{equation}
In addition, we note that
$$
\det[W(\lambda)] = -[W_{11}(\lambda)]^2 - W_{12}(\lambda) W_{21}(\lambda) = -\frac{b(\lambda)}{a(\lambda)}
$$
with
$$
b(\lambda) := \lambda^2 - \lambda_1^2 - 4 \lambda_1 \varphi_1 \varphi_2 - (\varphi_1^2+\varphi_2^2)^2 = \lambda^2 - \lambda_1^2 - E_0.
$$
Since $W_{12}(\lambda)$ has a simple zero at $\lambda = \mu$, then
\begin{equation}
\label{ham-1-1}
[W_{11}(\mu)]^2 = \frac{b(\mu)}{a(\mu)}.
\end{equation}
By substituting (\ref{ham-repr}), (\ref{ham-1-1}), and
$$
\frac{d}{dx} W_{12}(\mu) = -\frac{(\varphi_1^2+\varphi_2^2)}{a(\mu)} \frac{d\mu}{dx}
$$
to equation (\ref{ham-1-2}) and squaring it, we obtain the closed equation on $\mu$:
\begin{equation}
\label{ham-mu}
\frac{1}{4} \left( \frac{d\mu}{dx} \right)^2 = a(\mu) b(\mu) = (\mu^2 - \lambda_1^2)(\mu^2 - \lambda_1^2 - E_0).
\end{equation}
Substituting the representation (\ref{mu-representation}) yields
\begin{equation}
\label{u-representation}
u^2 \left( \frac{d^2u}{dx^2} \right)^2 - 2 u \left( \frac{du}{dx} \right)^2 \left[ \frac{d^2u}{dx^2} - 2(E_0+2\lambda_1^2) u \right]
= 16 \lambda_1^2 (E_0+\lambda_1^2) u^4.
\end{equation}
Let $u$ be a solution of the differential equations (\ref{ode}) with parameters $c$ and $d$.
Substituting (\ref{ode}) to (\ref{u-representation}) yields the relations
(\ref{ode-parameters}) between $(c,d)$ and $(\lambda_1,E_0)$. Hence, the constraint (\ref{u-representation})
is fulfilled  if $u$ satisfies (\ref{ode}) with parameters $(c,d)$ satisfying (\ref{ode-parameters}).
\end{proof}

\begin{proposition}
\label{prop-AKNS-time}
Let $u$, $\varphi = (\varphi_1,\varphi_2)^t$, and $\lambda_1$ be the same as in Proposition \ref{prop-AKNS}.
Then $\varphi(x-ct)$ satisfies the linear system (\ref{3.3}) with $\lambda = \lambda_1$ and $u(x-ct)$.
\end{proposition}

\begin{proof}
By using (\ref{ode}), we rewrite the first equation of system (\ref{3.3}) with $\lambda = \lambda_1$ as
\begin{equation}
\label{eq-varphi-1-time}
\partial_t \varphi_1 = - (4 \lambda_1^3 + 2 \lambda_1 u^2) \varphi_1 - (4 \lambda_1^2 u + 2 \lambda_1 u_x + c u) \varphi_2.
\end{equation}
By using (\ref{AKNS-solution}), we note that
\begin{eqnarray*}
(4 \lambda_1^2 + 2 u^2) \varphi_1 + (4 \lambda_1 u + 2 u_x) \varphi_2 =
(4 \lambda_1^2 + 2 u^2 + 8 \lambda_1 \varphi_1 \varphi_2) \varphi_1 = (4 \lambda_1^2 + 2 E_0) \varphi_1.
\end{eqnarray*}
By using (\ref{ode-parameters}) and the first equation in system (\ref{3.2}), equation (\ref{eq-varphi-1-time}) becomes
$$
\partial_t \varphi_1 = - \lambda_1 c \varphi_1 - c u \varphi_2 = -c \partial_x \varphi_1,
$$
hence $\varphi_1(x-ct)$ is a solution of system (\ref{3.2}) and (\ref{3.3}) with $\lambda = \lambda_1$ and $u(x-ct)$.
Similar computations hold for $\varphi_2$  by symmetry from the second equations in systems (\ref{3.2}) and
(\ref{3.3}).
\end{proof}

For the {\rm dn} Jacobian elliptic functions (\ref{elliptic-dn}), we have $c = 2 - k^2$ and $d = k^2 - 1 \leq 0$.
Since $u(x) > 0$ for every $x \in \mathbb{R}$, the periodic eigenfunction $\varphi = (\varphi_1,\varphi_2)^t$ in Proposition \ref{prop-AKNS}
is real with parameters $E_0 = \pm \sqrt{1 - k^2}$ and
\begin{equation}
\label{dn-branch-points-squared}
\lambda_1^2 = \frac{1}{4} \left[ 2 - k^2 \mp 2 \sqrt{1-k^2} \right].
\end{equation}
Taking the positive square root of (\ref{dn-branch-points-squared}), we obtain two
particular real points
\begin{equation}
\label{dn-branch-points}
\lambda_{\pm}(k) := \frac{1}{2} \left( 1 \pm \sqrt{1-k^2} \right),
\end{equation}
such that $0 < \lambda_-(k) < \lambda_+(k) < 1$ for every $k \in (0,1)$.
As $k \to 0$, we have $\lambda_-(k) \to 0$ and $\lambda_+(k) \to 1$,
whereas as $k \to 1$, we have $\lambda_-(k), \lambda_+(k) \to 1/2$.

For the {\rm cn} Jacobian elliptic functions (\ref{elliptic-cn}), we have $c = 2k^2-1$ and $d = k^2(1-k^2) \geq 0$.
Since $u(x)$ is sign-indefinite, the periodic eigenfunction $\varphi = (\varphi_1,\varphi_2)^t$ in Proposition \ref{prop-AKNS}
is complex-valued with parameters $E_0 = \pm i k \sqrt{1 - k^2}$ and
\begin{equation}
\label{cn-branch-points-squared}
\lambda_1^2 = \frac{1}{4} \left[ 2 k^2 - 1 \mp 2 i k \sqrt{1-k^2} \right].
\end{equation}
Defining the square root of (\ref{cn-branch-points-squared}) in the first quadrant of the complex plane,
we obtain
\begin{equation}
\label{cn-branch-points}
\lambda_I(k) := \frac{1}{2} \left( k + i \sqrt{1-k^2} \right).
\end{equation}
As $k \to 0$, we have $\lambda_I(k) \to i/2$, whereas as $k \to 1$, we have $\lambda_I(k) \to 1/2$.

Figure \ref{fig-3} shows the spectral plane of $\lambda$ with the schematic representation
of the Floquet--Bloch spectrum for the {\em dn}-periodic wave with $k = 0.75$ (left) and
the {\em cn}-periodic wave with $k = 0.75$ (right). The branch points $\lambda_{\pm}(k)$
and $\lambda_I(k)$ obtained in (\ref{dn-branch-points}) and (\ref{cn-branch-points}) are
marked explicitly as the end points of the Floquet--Bloch spectral bands away from the imaginary axis.

\begin{figure}[htb]
\begin{center}
\includegraphics[height=4.1cm]{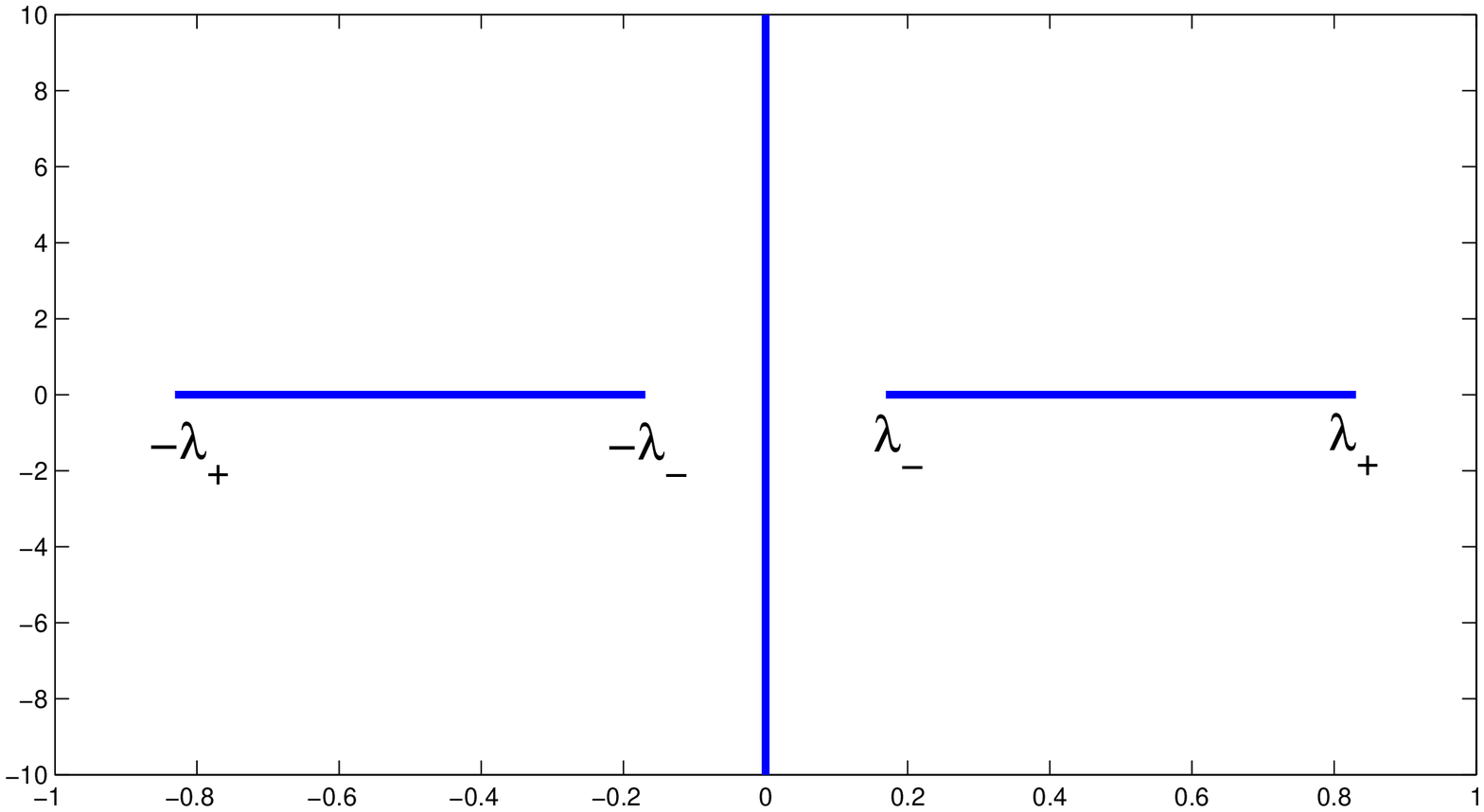}
\includegraphics[height=4.1cm]{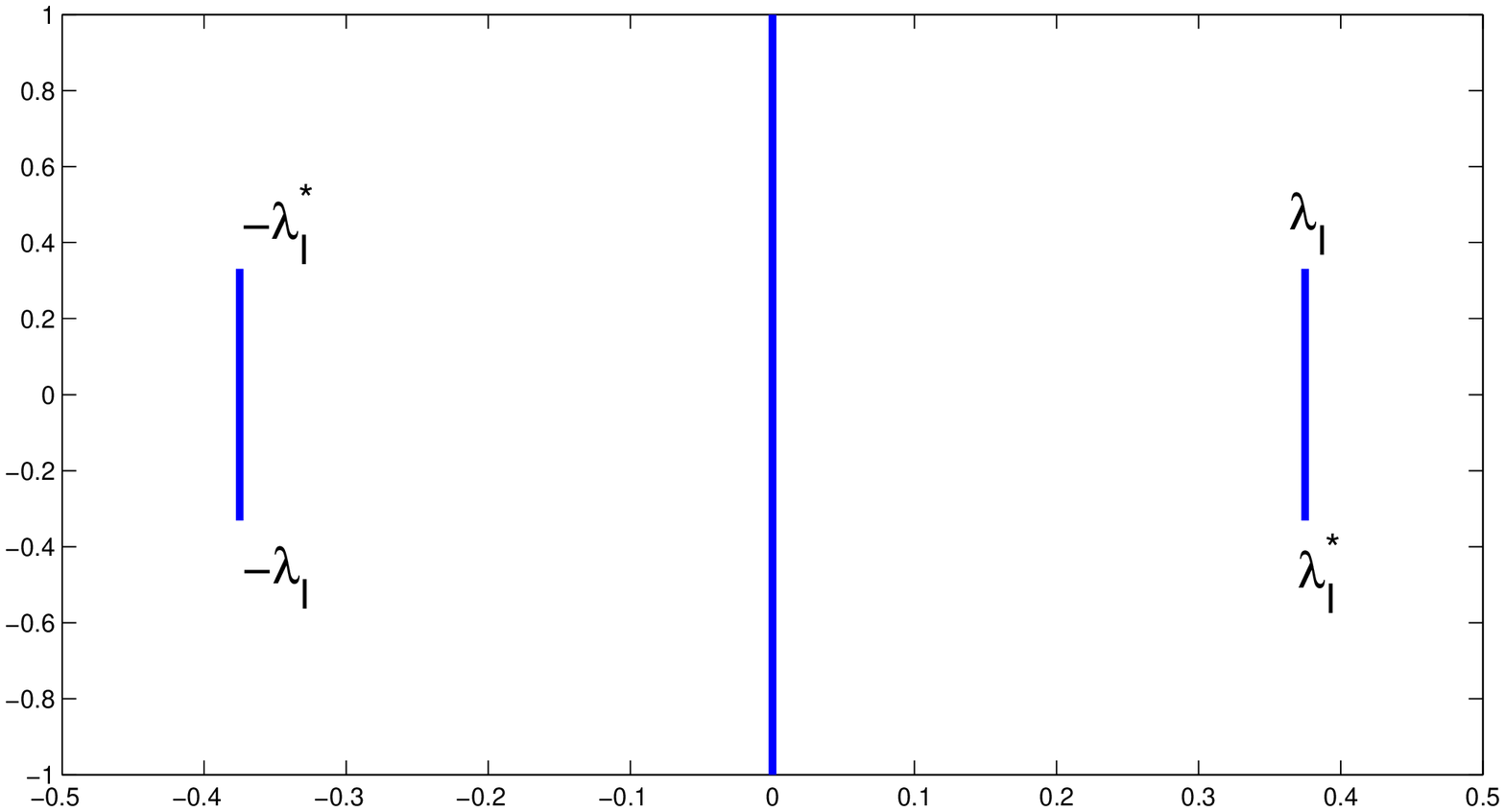}
\end{center}
\caption{The spectral plane for the {\em dn}-periodic (left) and
{\em cn}-periodic (right) waves with $k = 0.75$.}
\label{fig-3}
\end{figure}

\section{Non-periodic solutions of the AKNS spectral problem}

Here we construct the second linearly independent solution
to the AKNS spectral problem (\ref{3.2}) with $\lambda = \lambda_1$ and extend it to satisfy
the linear system (\ref{3.3}). The second solution
is no longer periodic in variables $(x,t)$. The following result represents the corresponding solution.

\begin{proposition}
\label{prop-Second-solution}
Let $u$, $\lambda_1$, $E_0$, and $\varphi = (\varphi_1,\varphi_2)^t$ be the same as in Proposition \ref{prop-AKNS}.
Assume that $u(x)^2 - E_0 \neq 0$ for every $x$.
The second linearly independent solution of the AKNS spectral problem (\ref{3.2}) with $\lambda = \lambda_1$ is given by
$\psi = (\psi_1,\psi_2)^t$, where
\begin{equation}
\label{Wronskian-repr}
\psi_1 = \frac{\theta - 1}{\varphi_2}, \quad \psi_2 = \frac{\theta + 1}{\varphi_1},
\end{equation}
and
\begin{equation}
\label{theta-repr}
\theta(x) = -4 \lambda_1 (u(x)^2-E_0) \int_0^x\frac{u(y)^2}{(u(y)^2-E_0)^2}dy.
\end{equation}
In particular, if $u$ is periodic in $x$, then $\theta$ grows linearly in $x$ as $|x| \to \infty$,
so that $\psi_1$ and $\psi_2$ are not periodic in $x$.
\end{proposition}

\begin{proof}
Since the AKNS spectral problem (\ref{3.2}) is related to the traceless matrix, the Wronskian
of the two linearly independent solutions $\varphi = (\varphi_1,\varphi_2)^t$ and $\psi = (\psi_1,\psi_2)^t$
is independent of $x$. Normalizing it by $2$, we write the relation
$$
\varphi_1 \psi_2 - \varphi_2 \psi_1 = 2,
$$
from which the representation (\ref{Wronskian-repr}) follows with arbitrary $\theta$. If $u(x)^2 - E_0 \neq 0$
for every $x$, then $\varphi_1(x) \neq 0$ and $\varphi_2(x) \neq 0$ for every $x$.
Substituting (\ref{Wronskian-repr})
to (\ref{3.2}), we obtain the following scalar linear differential equation for $\theta$:
\begin{eqnarray*}
\frac{d \theta}{d x} = u \theta \frac{\varphi_2^2-\varphi_1^2}{\varphi_1 \varphi_2} + u \frac{\varphi_1^2+\varphi_2^2}{\varphi_1 \varphi_2}.
\end{eqnarray*}
By using relations (\ref{AKNS-solution}), we rewrite it in the equivalent forms:
\begin{eqnarray*}
\frac{d \theta}{d x} = \theta \frac{2 u u'}{u^2 - E_0} - \frac{4 \lambda_1 u^2}{u^2 - E_0} \quad \Rightarrow \quad
\frac{d}{dx} \left[ \frac{\theta}{u^2 - E_0} \right] = - \frac{4 \lambda_1 u^2}{(u^2 - E_0)^2}.
\end{eqnarray*}
Integrating the last equation with the boundary condition $\theta(0) = 0$, we obtain (\ref{theta-repr}).
\end{proof}

\begin{proposition}
\label{prop-Second-solution-time}
Let $u$, $\lambda_1$, $E_0$, $\varphi = (\varphi_1,\varphi_2)^t$, and $\psi = (\psi_1,\psi_2)^t$ be the same as
in Proposition \ref{prop-Second-solution}.
Then, $\psi = (\psi_1,\psi_2)^t$ expressed by (\ref{Wronskian-repr}) satisfies the linear system (\ref{3.3})
with $\lambda = \lambda_1$ and $u(x-ct)$ if $\theta$ is expressed by
\begin{equation}
\label{theta-repr-tim}
\theta(x,t) = -4 \lambda_1 (u(x-ct)^2-E_0) \left[ \int_0^{x-ct} \frac{u(y)^2}{(u(y)^2-E_0)^2} dy - t\right].
\end{equation}
\end{proposition}

\begin{proof}
By using (\ref{ode}), we rewrite the first equation of system (\ref{3.3}) with $\lambda = \lambda_1$ as
\begin{equation}
\label{eq-psi-1-time}
\partial_t \psi_1 = - (4 \lambda_1^3 + 2 \lambda_1 u^2) \psi_1 - (4 \lambda_1^2 u + 2 \lambda_1 u_x + c u) \psi_2.
\end{equation}
By using (\ref{AKNS-solution}), (\ref{Wronskian-repr}),
and expressing $\partial_t \varphi_2$ from the second equation of system (\ref{3.3}),
we obtain from (\ref{eq-psi-1-time}):
\begin{eqnarray*}
\partial_t \theta & = & \frac{(4 \lambda_1^2 u - 2 \lambda_1 u_x + c u) \varphi_1 (\theta - 1)}{\varphi_2}
- \frac{(4 \lambda_1^2 u + 2 \lambda_1 u_x + c u) \varphi_2 (\theta + 1)}{\varphi_1} \\
& = & - 16 \lambda_1^2 \varphi_1 \varphi_2 - \frac{c u}{\varphi_1 \varphi_2}
\left[ \theta (\varphi_2^2-\varphi_1^2) + \varphi_1^2 + \varphi_2^2 \right] \\
& = & 4 \lambda_1 (u^2 - E_0) - c \partial_x \theta.
\end{eqnarray*}
Let us represent $\theta = -4 \lambda_1 (u^2 - E_0) \chi$ so that $\chi$ satisfies
$$
\partial_t \chi =  - c \partial_x \chi - 1.
$$
Hence $\chi(x,t) = - t + f(x-ct)$, where $f$ is obtained from (\ref{theta-repr}) in the form
$$
f(x) = \int_0^{x} \frac{u(y)^2}{(u(y)^2-E_0)^2} dy
$$
to yield the representation (\ref{theta-repr-tim}). Similar computations hold for $\psi_2$ by symmetry
from the second equations in systems (\ref{3.2}) and (\ref{3.3}).
\end{proof}

Note that a more general solution for $\psi = (\psi_1,\psi_2)^t$
is defined arbitrary up to an addition to the first solution $\varphi = (\varphi_1,\varphi_2)^t$. However,
this addition is equivalent to the arbitrary choice of the lower limit in the integral (\ref{theta-repr-tim}),
which is then equivalent to the translation in time $t$. Thus, the second linearly independent solution
in the form (\ref{Wronskian-repr}) and (\ref{theta-repr-tim}) is unique up to the translation
in $x$ and $t$.

\section{One-fold, two-fold, and $N$-fold Darboux transformations}

Here we give the explicit formulas for the one-fold and two-fold Darboux transformations for the focusing mKdV equation (\ref{mKdV}),
as well as the general formula for the $N$-fold Darboux transformation.
Although the formal derivation of the $N$-fold Darboux transformation
can be found in several sources, e.g. in book \cite{GuHuZhou}
or original papers \cite{GengTam,WangChen}, we find it useful to derive
the explicit transformation formulas by using purely algebraic calculations.

By definition, we say that $T(\lambda)$ is a Darboux transformation if
\begin{equation}
\label{Darboux}
\widetilde{\varphi} = T(\lambda) \varphi,
\end{equation}
where $\varphi$ satisfies (\ref{3.2})--(\ref{3.3}) for a particular potential $u$ and
$\widetilde{\varphi}$ satisfies (\ref{3.2})--(\ref{3.3}) for a new potential $\widetilde{u}$,
which is related to $u$. The transformation formulas between $\varphi$ and $\widetilde{\varphi}$
follow from the Darboux equations
\begin{equation}
\label{darboux-1}
\partial_x T(\lambda) + T(\lambda) U(\lambda,u) = U(\lambda,\widetilde{u}) T(\lambda)
\end{equation}
and
\begin{equation}
\label{darboux-2}
\partial_t T(\lambda) + T(\lambda) V(\lambda,u) = V(\lambda,\widetilde{u}) T(\lambda).
\end{equation}

In many derivations, e.g. in \cite{GengTam,GuHuZhou,WangChen}, the $N$-fold Darboux transformation is deduced
formally from a linear system of equations imposed on the coefficients of the polynomial representation
of $T(\lambda)$ without checking all the constraints arising from the Darboux equations (\ref{darboux-1}) and (\ref{darboux-2}).
In order to avoid such formal computations, we give in Appendix A a rigorous derivation of the $N$-fold Darboux transformation
in the explicit form and show how the Darboux equations (\ref{darboux-1}) and (\ref{darboux-2}) are satisfied.
Our derivation relies on a particular implementation of the dressing method \cite{ZS1,ZS2} which was
recently reviewed in the context of the cubic NLS equation in \cite{Contreras}.

The general $N$-fold Darboux transformation is given by the following theorem.

\begin{theorem}
\label{theorem-N-fold}
Let $u$ be a smooth solution of the mKdV equation (\ref{mKdV}). Let $\varphi^{(k)} = (p_k,q_k)^t$, $1 \leq k \leq N$ be a particular
smooth nonzero solution of system (\ref{3.2}) and (\ref{3.3}) with fixed $\lambda = \lambda_k \in \mathbb{C} \backslash \{0\}$
and potential $u$. Assume that $\lambda_k \neq \pm \lambda_j$ for every $k \neq j$.
Let $\{\widetilde{\varphi}^{(k)}\}_{1 \leq k \leq N}$ be a solution of the linear algebraic system
\begin{equation}
\label{lin-system}
\sigma_3 \sigma_1 \varphi^{(j)} = \sum_{k=1}^N \frac{\langle \varphi^{(j)}, \varphi^{(k)} \rangle}{\lambda_j + \lambda_k} \widetilde{\varphi}^{(k)}, \quad 1 \leq j \leq N,
\end{equation}
where $\langle \varphi^{(j)}, \varphi^{(k)} \rangle := p_j p_k + q_j q_k$ is the inner vector product.
Assume that the linear system (\ref{lin-system}) has a unique solution. Then,
$\widetilde{\varphi}^{(k)} = (\widetilde{p}_k,\widetilde{q}_k)^t$, $1 \leq k \leq N$ is a particular
solution of system (\ref{3.2}) and (\ref{3.3})
with $\lambda = \lambda_k$ and the new potential $\widetilde{u}$ given by
\begin{equation}
\label{N-fold}
\widetilde{u} = u + 2 \sum_{j=1}^N \widetilde{p}_j p_j = u - 2 \sum_{j=1}^N \widetilde{q}_j q_j.
\end{equation}
Consequently, $\widetilde{u}$ is a new solution of the mKdV equation (\ref{mKdV}).
\end{theorem}

The proof of Theorem \ref{theorem-N-fold} is given in Appendix A.
The following two propositions represent the one-fold and two-fold Darboux transformation formulas
deduced from Theorem \ref{theorem-N-fold} for $N = 1$ and $N = 2$ respectively.

\begin{proposition}\label{prop-one-fold}
Let $u$ be a smooth solution of the mKdV equation (\ref{mKdV}). Let $\varphi = (p,q)^t$ be a particular
smooth nonzero  solution of system (\ref{3.2}) and (\ref{3.3}) with fixed $\lambda = \lambda_1 \in \mathbb{C} \backslash \{0\}$.
Then,
\begin{equation}
\label{one-fold}
\widetilde{u} = u + \frac{4 \lambda_1 p q}{p^2+q^2}
\end{equation}
is a new solution of the mKdV equation (\ref{mKdV}).
\end{proposition}

\begin{proof}
Solving the linear system (\ref{lin-system}) for $\widetilde{\varphi} = (\widetilde{p},\widetilde{q})^t$ yields
\begin{equation}
\label{symb-9}
\widetilde{p} = \frac{2 \lambda_1 q}{p^2+q^2}, \quad \widetilde{q} = \frac{-2 \lambda_1 p}{p^2+q^2}.
\end{equation}
Substituting (\ref{symb-9}) into (\ref{N-fold}) for $N = 1$ results in the transformation formula (\ref{one-fold}).
\end{proof}

\begin{proposition}\label{prop-two-fold}
Let $u$ be a smooth solution of the mKdV equation (\ref{mKdV}). Let $\varphi^{(k)} = (p_k,q_k)^t$ be a particular
smooth nonzero solution of system (\ref{3.2}) and (\ref{3.3}) with fixed $\lambda = \lambda_k \in \mathbb{C} \backslash \{0\}$
for $k = 1,2$ such that $\lambda_1 \neq \pm \lambda_2$. Then,
\begin{equation}
\label{two-fold}
\tilde{u} = u + \frac{4 (\lambda_1^2-\lambda_2^2)
\left[  \lambda_1 p_1 q_1 (p_2^2 + q_2^2) - \lambda_2 p_2 q_2 (p_1^2 + q_1^2) \right]}{
(\lambda_1^2 + \lambda_2^2) (p_1^2+q_1^2) (p_2^2+q_2^2)
- 2 \lambda_1 \lambda_2 \left[4 p_1 q_1 p_2 q_2 + (p_1^2-q_1^2)(p_2^2-q_2^2) \right]}
\end{equation}
is a new solution of the mKdV equation (\ref{mKdV}).
\end{proposition}

\begin{proof}
The linear system (\ref{lin-system}) is generated by the matrix $A$ with the entries
\begin{equation}
\label{entry-A}
A_{jk} = \frac{\langle \varphi^{(j)}, \varphi^{(k)} \rangle}{\lambda_j + \lambda_k}, \quad 1 \leq j,k \leq N.
\end{equation}
For $N = 2$, we compute the determinant of this matrix as
\begin{eqnarray*}
{\rm det}(A) & = & \frac{1}{4 \lambda_1 \lambda_2 (\lambda_1+\lambda_2)^2} \left[
(\lambda_1 + \lambda_2)^2 (p_1^2+q_1^2)(p_2^2+q_2^2) - 4 \lambda_1 \lambda_2 (p_1p_2 +q_1q_2)^2\right] \\
& = & \frac{1}{4 \lambda_1 \lambda_2 (\lambda_1+\lambda_2)^2} \left[
(\lambda_1^2 + \lambda_2^2) (p_1^2+q_1^2)(p_2^2+q_2^2)
- 2 \lambda_1 \lambda_2 \left( 4 p_1 p_2 q_1 q_2 + (p_1^2-q_1^2)(p_2^2 - q_2^2) \right) \right].
\end{eqnarray*}
Solving the linear system (\ref{lin-system}) with Cramer's rule yields the components
$$
\widetilde{p}_1 = \frac{(\lambda_1 + \lambda_2) q_1 (p_2^2+q_2^2)-2 \lambda_2 q_2 (p_1p_2+q_1q_2)}{2 \lambda_2 (\lambda_1+\lambda_2) {\rm det}(A)}
$$
and
$$
\widetilde{p}_2 = \frac{(\lambda_1 + \lambda_2) q_2 (p_1^2+q_1^2)-2 \lambda_1 q_1 (p_1p_2+q_1q_2)}{2 \lambda_1 (\lambda_1+\lambda_2) {\rm det}(A)}.
$$
Substituting these formulas to the representation (\ref{N-fold}) with $N = 2$ and reordering the similar terms result in
the transformation formula (\ref{two-fold}).
\end{proof}

\section{The ``rogue" {\em dn}-periodic wave}

Here we apply the one-fold Darboux transformation (\ref{one-fold}) to the Jacobian elliptic function
{\em dn} in (\ref{elliptic-dn}) in order to obtain the ``rogue" {\em dn}-periodic wave. We write the ``rogue" wave in
commas, because the corresponding solution is a nonlinear superposition of an algebraically decaying soliton and the {\em dn}-periodic
wave, hence the maximal amplitude is brought by the algebraic soliton from infinity.
The proper rogue {\em dn}-periodic wave does not exist in the mKdV equation (\ref{mKdV}) because
the {\em dn}-periodic wave is modulationally stable. We note however that very similar solutions to the NLS
equation define a proper rogue {\em dn}-periodic wave, as is shown numerically in \cite{Kedziora}.

Let $u$ be the {\em dn} periodic wave (\ref{elliptic-dn}), whereas $\varphi = (\varphi_1,\varphi_2)^t$ be the periodic
eigenfunction of the linear system (\ref{3.2}) and (\ref{3.3}) with $\lambda = \lambda_1$
defined by Propositions \ref{prop-AKNS} and \ref{prop-AKNS-time}.
Since the connection formulas (\ref{AKNS-solution}) are satisfied for every $t \in \mathbb{R}$,
substituting $p = \varphi_1$ and $q = \varphi_2$ into the one-fold Darboux transformation (\ref{one-fold}) yields
another solution of the mKdV equation in the form
$$
\tilde{u} = u + \frac{4 \lambda_1 \varphi_1 \varphi_2}{\varphi_1^2 + \varphi_2^2} = \frac{E_0}{u},
$$
where $E_0 = \pm \sqrt{1 - k^2}$. However, since
$$
{\rm dn}(x + K(k);k) = \frac{\sqrt{1-k^2}}{{\rm dn}(x;k)},
$$
the new solution $\tilde{u}$ to the mKdV equation (\ref{mKdV}) is obtained trivially by the spatial translation
of the {\em dn} periodic wave on the half-period $\frac{1}{2} L = K(k)$. This computation explains why
we need to use the second non-periodic solution $\psi$ instead of the periodic eigenfunction $\varphi$.

Let $u$ be the {\em dn} periodic wave (\ref{elliptic-dn}), whereas $\psi = (\psi_1,\psi_2)^t$
be the non-periodic solution to the linear system (\ref{3.2}) and (\ref{3.3}) with $\lambda = \lambda_1$
defined by Propositions \ref{prop-Second-solution} and \ref{prop-Second-solution-time}.
Recall that there exist two choices for $\lambda_1$ in (\ref{dn-branch-points}).
However, for the choice $\lambda_1 = \lambda_-(k)$, we have $E_0 = \sqrt{1-k^2}$ and $u(x)^2 - E_0 = 0$
for some values of $x$ in $[-K(k),K(k)]$,
therefore, the assumption of Proposition \ref{prop-Second-solution} is not satisfied.
For the choice $\lambda_1 = \lambda_+(k)$, we have $E_0 = -\sqrt{1-k^2}$ and $u(x)^2 - E_0 > 0$ for every $x$,
therefore, the assumption of Proposition \ref{prop-Second-solution} is satisfied. Substituting
$p = \psi_1$ and $q = \psi_2$ given by (\ref{Wronskian-repr}) into the one-fold Darboux transformation (\ref{one-fold})
with $\lambda_1 = \lambda_+(k)$ and $E_0 = -\sqrt{1-k^2}$ yields another solution of the mKdV equation in the form
$$
\tilde{u} = u + \frac{4 \lambda_1 \psi_1 \psi_2}{\psi_1^2 + \psi_2^2} =
u + \frac{4 \lambda_1 \varphi_1 \varphi_2 (\theta^2 -1)}{(\varphi_1^2+\varphi_2^2) (1 + \theta^2) - 2 (\varphi_1^2-\varphi_2^2) \theta}.
$$
By using relations (\ref{AKNS-solution}) again, we finally write the new solution in the form
\begin{equation}
\label{dn-rogue}
u_{\rm dn-rogue} = u_{\rm dn} + \frac{(1-\theta_{\rm dn}^2) (u_{\rm dn}^2+\sqrt{1-k^2})}{
(1+\theta_{\rm dn}^2)u_{\rm dn} - \lambda_1^{-1} \theta_{\rm dn} u_{\rm dn}'}
\end{equation}
where
\begin{equation}
\label{dn-rogue-theta}
\theta_{\rm dn}(x,t) = -4 \lambda_1 (u_{\rm dn}(x-ct)^2+\sqrt{1-k^2})
\left[ \int_0^{x-ct} \frac{u_{\rm dn}(y)^2}{(u_{\rm dn}(y)^2+ \sqrt{1-k^2})^2}dy - t \right].
\end{equation}
We refer to the exact solution (\ref{dn-rogue})--(\ref{dn-rogue-theta}) as the ``rogue"
{\em dn} periodic wave of the mKdV equation.

If $k = 0$, then $u_{\rm dn}(x,t) = 1$, $\lambda_1 = 1$, $c = 2$, $\theta_{\rm dn}(x,t) = -2(x-6t)$, and
$$
k = 0 : \quad u_{\rm dn-rogue}(x,t) = -1 + \frac{4}{1 + 4 (x-6t)^2}.
$$
Although this expression is an analogue of the rogue wave of the NLS on the constant wave background \cite{Akh,Choudhury},
it corresponds to the algebraically decaying soliton of the mKdV \cite{GrimshawPel}.

If $k = 1$, then $u_{\rm dn}(x,t) = {\rm sech}(x-t)$, $\lambda_1 = \frac{1}{2}$, $c = 1$,
$$
\theta_{\rm dn}(x,t) = -(x - 3t) {\rm sech}^2(x-t) - \tanh(x-t),
$$
and
$$
k = 1 : \quad u_{\rm dn-rogue}(x,t) = 2 {\rm sech}(x-t) \frac{1 - (x-3t) \tanh(x-t)}{1 + (x-3t)^2 {\rm sech}^2(x-t)}.
$$
in agreement with the two-soliton solutions of the mKdV for two nearly identical solitons \cite{PelSh,PelSl}.

Next, we show that for every $k \in [0,1)$, there exists a particular line $x = c_* t$ with $c_* > c$ such that
$\theta_{\rm dn}(x,t)$ given by (\ref{dn-rogue-theta}) remains bounded as $|x| + |t| \to \infty$.
This value of $c_*$ gives the speed of the algebraically decaying soliton propagating on the {\em dn}-periodic
wave background. For instance, if $k = 0$, then $c_* = 6 > 2 = c$.

In order to show the claim above, we inspect the expression
$$
\int_0^{x-ct} \frac{u_{\rm dn}(y)^2}{(u_{\rm dn}(y)^2+ \sqrt{1-k^2})^2}dy - t.
$$
Since the integrand is a positive $L=2K(k)$-periodic function with a positive mean value denoted by $I(k)$,
then the expression can be written as
$$
I(k) (x-ct) - t + \mbox{\rm a periodic function of $(x,t)$}.
$$
Therefore, $\theta_{\rm dn}(x,t)$ is bounded at $x = c_* t$, where $c_* = c + [I(k)]^{-1} > c$.

Except for the line $x = c_* t$, the function $\theta_{\rm dn}(x,t)$ given by (\ref{dn-rogue-theta})
grows linearly in $x$ and $t$ as $|x| + |t| \to \infty$ for every $k \in [0,1)$. Hence
the representation (\ref{dn-rogue}) yields asymptotic behavior
\begin{equation*}
u_{\rm dn-rogue}(x,t) \sim - \frac{\sqrt{1-k^2}}{{\rm dn}(x-ct;k)} = -{\rm dn}(x-ct+K(k);k) = -u_{\rm dn}(x-ct+K(k)).
\end{equation*}
The maximal value of $u_{\rm dn-rogue}(x,t)$ as $|x| + |t| \to \infty$ except for the line $x = c_* t$
coincides with the maximal value of $u_{\rm dn}(x,t) = {\rm dn}(x-ct;k)$.

For $t = 0$, $u_{\rm dn}(x,0)$ is even in $x$, $\theta_{\rm dn}(x,0)$ is odd in $x$, hence
$u_{\rm dn-rogue}(x,0)$ is even in $x$. The maximal value of $u_{\rm dn}(x,0)$ occurs at $u_{\rm dn}(0,0) = 1$.
Since $u_{\rm dn-rogue}(x,0)$ is even in $x$, then $x = 0$ is an extremal
point of $u_{\rm dn-rogue}(x,0)$. Moreover, $\partial_x^2 u_{\rm dn-rogue}(0,0) < 0$, which follows
from the expansions $u_{\rm dn}(x,0) = 1 - \frac{1}{2} k^2 x^2 + \mathcal{O}(x^4)$,
$\theta_{\rm dn}(x,0) = -4 \lambda_1 (1+\sqrt{1-k^2})^{-1} x + \mathcal{O}(x^3)$, and
$$
u_{\rm dn-rogue}(x,0) = 2 + \sqrt{1-k^2} - \left[ 8 - 3 k^2 + 8 \sqrt{1-k^2} - \frac{1}{2} k^2 \sqrt{1-k^2} \right] x^2 + \mathcal{O}(x^4).
$$
Hence $x = 0$ is the point of maximum of $u_{\rm dn-rogue}(x,0)$. Defining the magnification number as
$$
M_{\rm dn}(k) = \frac{u_{\rm dn-rogue}(0,0)}{\max\limits_{x \in [-K(k),K(k)]} u_{\rm dn}(x,0)} = 2 + \sqrt{1-k^2},
$$
we obtain the expression in (\ref{magnification-factor}). The value $M_{\rm dn}(k)$ corresponds to the amplitude
of the algebraically decaying soliton propagating on the background of the {\em dn}-periodic wave.

\section{The rogue {\em cn}-periodic wave}

Here we apply the one-fold and two-fold Darboux transformations (\ref{one-fold}) and (\ref{two-fold})
to the Jacobian elliptic function {\em cn} in (\ref{elliptic-cn}) in order to obtain the rogue {\em cn}-periodic wave.
This is a proper rogue {\em cn}-periodic wave because
the {\em cn}-periodic wave is modulationally unstable.

Let $u$ be the {\em cn} periodic wave (\ref{elliptic-cn}), whereas $\varphi = (\varphi_1,\varphi_2)^t$ be the periodic solution
to the linear system (\ref{3.2}) and (\ref{3.3}) with $\lambda = \lambda_1$
defined by Propositions \ref{prop-AKNS} and \ref{prop-AKNS-time}. Without loss of generality,
we choose $\lambda_1 = \lambda_I(k)$, where $\lambda_I(k)$ is given by (\ref{cn-branch-points}),
so that $E_0 = - i k \sqrt{1-k^2}$. Since the periodic solution $\varphi$ is complex,
the one-fold Darboux transformation (\ref{one-fold}) produces a complex-valued solution to the mKdV,
hence we should use the two-fold Darboux transformation (\ref{two-fold}).

By virtue of relations (\ref{AKNS-solution}), substituting $(p_1,q_1) = (\varphi_1,\varphi_2)$ with $\lambda_1 = \lambda_I$
and $(p_2,q_2) = (\overline{\varphi}_1,\overline{\varphi}_2)$ with $\lambda_2 = \overline{\lambda}_I$
to the two-fold Darboux transformation (\ref{two-fold}) yields
another solution of the mKdV equation in the form
$$
\tilde{u} = u + \frac{4 k^2 (1-k^2) u}{(2k^2-1) u^2 - u^4 - k^2(1-k^2) - (u')^2} = -u,
$$
where the first-order invariant in (\ref{ode}) is used in the second identity with $c = 2k^2-1$ and $d = k^2(1-k^2)$.
Thus, the new solution $\tilde{u}$ in the two-fold transformation (\ref{two-fold}) is trivially related
to the previous solution $u$ if the functions $(p_1,q_1)$ and $(p_2,q_2)$ are periodic.

Let us now consider the non-periodic solution $\psi = (\psi_1,\psi_2)^t$
to the linear system (\ref{3.2}) and (\ref{3.3}) with $\lambda = \lambda_I$.
The assumption of Proposition \ref{prop-Second-solution} is satisfied
because $E_0 = - i k \sqrt{1-k^2} \neq 0$ for $k \in (0,1)$ and $u(x)^2 - E_0 \neq 0$ for every $x$. Therefore,
the non-periodic solution $\psi$ in  Propositions \ref{prop-Second-solution} and \ref{prop-Second-solution-time}
is well-defined. Substituting $(p_1,q_1) = (\psi_1,\psi_2)$ with $\lambda_1 = \lambda_I$
and $(p_2,q_2) = (\overline{\psi}_1,\overline{\psi}_2)$ with $\lambda_2 = \overline{\lambda}_I$
into the two-fold Darboux transformation (\ref{two-fold}) yields another solution of the mKdV in the form
\begin{eqnarray*}
\tilde{u} = u + \frac{4 (\lambda_I^2-\overline{\lambda}_I^2) \left[
\lambda_I \psi_1 \psi_2 (\overline{\psi}_1^2 + \overline{\psi}_2^2) -
\overline{\lambda}_I \overline{\psi}_1 \overline{\psi}_2 (\psi_1^2 + \psi_2^2) \right]}{
(\lambda_I^2 + \overline{\lambda}_I^2) |\psi_1^2+\psi_2^2|^2
- 2 |\lambda_I|^2 \left[4 |\psi_1|^2 |\psi_2|^2 + |\psi_1^2-\psi_2^2|^2  \right]}
= u + \frac{F_1}{F_2},
\end{eqnarray*}
where
\begin{eqnarray*}
F_1 & = & 8 {\rm Im}(\lambda_I^2) {\rm Im}\left[ \lambda_I \varphi_1 \varphi_2 (1-\theta^2) [(1+\overline{\theta}^2)(\overline{\varphi}_1^2+\overline{\varphi}_2^2) -
2 \overline{\theta} (\overline{\varphi}_1^2-\overline{\varphi}_2^2)] \right], \\
F_2 & = & {\rm Re}(\lambda_I^2) |(1+\theta^2)(\varphi_1^2+\varphi_2^2) - 2 \theta (\varphi_1^2-\varphi_2^2)|^2 \\
& \phantom{t} & - |\lambda_I|^2 \left( 4 |1-\theta^2|^2 |\varphi_1|^2 |\varphi_2|^2 + |(1+\theta^2)(\varphi_1^2-\varphi_2^2)
- 2 \theta (\varphi_1^2+\varphi_2^2)|^2 \right).
\end{eqnarray*}
By using relations (\ref{AKNS-solution}) and (\ref{cn-branch-points}), we finally write the new solution in the form
\begin{equation}
\label{cn-rogue}
u_{\rm cn-rogue} = u_{\rm cn} + \frac{G_1}{G_2},
\end{equation}
where
\begin{eqnarray*}
G_1 & = & 4 k \sqrt{1-k^2} {\rm Im}\left[ (u_{\rm cn}^2+i k \sqrt{1-k^2}) (1-\theta_{\rm cn}^2)
[(1+\overline{\theta}_{\rm cn}^2) u_{\rm cn} - \overline{\lambda}_I^{-1} \overline{\theta}_{\rm cn} u_{\rm cn}'] \right], \\
G_2 & = & (1-2k^2) |(1+\theta_{\rm cn}^2) u_{\rm cn} - \lambda_I^{-1} \theta_{\rm cn} u_{\rm cn}'|^2 \\
& \phantom{t} & + |1-\theta_{\rm cn}^2|^2 \left[ u_{\rm cn}^4 + k^2(1-k^2) \right] +
|(1+\theta_{\rm cn}^2) (2\lambda_I)^{-1} u_{\rm cn}' - 2 \theta_{\rm cn} u_{\rm cn}|^2,
\end{eqnarray*}
and
\begin{equation}
\label{cn-rogue-theta}
\theta_{\rm cn}(x,t) = -4 \lambda_I (u_{\rm cn}(x-ct)^2+i k \sqrt{1-k^2})
\left[ \int_0^{x-ct} \frac{u_{\rm cn}(y)^2}{(u_{\rm cn}(y)^2+ i k\sqrt{1-k^2})^2}dy - t\right].
\end{equation}
We refer to the exact solution (\ref{cn-rogue})--(\ref{cn-rogue-theta})
as the rogue {\em cn} periodic wave of the mKdV equation.

As $k \to 0$, then $u_{\rm cn}(x,t) \to 0$, $\lambda_I \to i/2$, $\theta_{\rm cn}(x,t) \to 0$, and
$u_{\rm cn-rogue}(x,t) \to 0$. Although the limit is zero, one can derive asymptotic expansions at the
order of $\mathcal{O}(k)$ which recovers the rogue wave of the NLS equation (\ref{nls}),
according to the asymptotic transformation of the focusing mKdV to the focusing NLS in the small-amplitude
limit \cite{PelPel}. The rouge {\em cn}-periodic wave generalizes the rogue wave (\ref{rogue-basic})
on the constant wave background.

As $k \to 1$, then $u_{\rm cn}(x,t) \to {\rm sech}(x-t)$, $\lambda_I \to 1/2$,
and it may first seem that the second term in (\ref{cn-rogue}) vanishes.
However, $G_1 = \mathcal{O}(1-k^2)$ and $G_2 = \mathcal{O}(1-k^2)$,
hence a non-trivial limit exists to yield a three-soliton solution to the mKdV with three nearly identical
solitons \cite{PelSl}.

Let us inspect the expression
$$
\int_0^{x-ct} \frac{u_{\rm cn}(y)^2}{(u_{\rm cn}(y)^2+ i k\sqrt{1-k^2})^2}dy - t =
\int_0^{x-ct} \frac{u_{\rm cn}(y)^2 (u_{\rm cn}(y)^2- i k \sqrt{1-k^2})^2}{(u_{\rm cn}(y)^4+ k^2 (1-k^2))^2}dy - t.
$$
For every $k \in (0,1)$, the imaginary part in the integrand is a negative $L=4K(k)$-periodic function with a negative
mean value. It is only bounded on the line $x = ct$, however, the real part of the last term in the expression grows
linearly in $t$. Therefore, for every $k \in (0,1)$, $|\theta_{\rm cn}(x,t)|$ grows linearly in $x$ and $t$ as
$|x| + |t| \to \infty$ everywhere on the $(x,t)$ plane. Hence the representation (\ref{cn-rogue})
yields the asymptotic behavior
\begin{eqnarray*}
u_{\rm cn-rogue}(x,t) & \sim & u_{\rm cn}(x,t) + \frac{4 k^2 (1-k^2) u_{\rm cn}(x,t)}{
(2k^2-1) u_{\rm cn}(x,t)^2 - (\partial_x u_{\rm cn}(x,t))^2 - u_{\rm cn}(x,t)^4-k^2(1-k^2)}\\
&  = & -u_{\rm cn}(x,t),
\end{eqnarray*}
where the first-order invariant in (\ref{ode}) is used for the last identity
with $c = 2k^2-1$ and $d = k^2(1-k^2)$.
The maximal value of $u_{\rm cn-rogue}(x,t)$ as $|x| + |t| \to \infty$
coincides with the maximal value of $u_{\rm cn}(x,t) = k {\rm cn}(x-ct;k)$.

For $t = 0$, $u_{\rm cn}(x,0)$ is even in $x$, $\theta_{\rm cn}(x,0)$ is odd in $x$, hence
$u_{\rm cn-rogue}(x,0)$ is even in $x$. The maximal value of $u_{\rm cn}(x,0)$ occurs at $u_{\rm cn}(0,0) = k$.
Since $u_{\rm cn-rogue}(x,0)$ is even in $x$, then $x = 0$ is an extremal
point of $u_{\rm cn-rogue}(x,0)$. Moreover, $\partial_x^2 u_{\rm cn-rogue}(0,0) < 0$, which follows
from the expansions $u_{\rm cn}(x,0) = k - \frac{1}{2} k x^2 + \mathcal{O}(x^4)$,
$\theta_{\rm cn}(x,0) = -4 \lambda_I (k^2 - i k \sqrt{1-k^2}) x + \mathcal{O}(x^3)$, and
$$
u_{\rm cn-rogue}(x,0) = 3k - \left[ \frac{3}{2} k + 16 k^3 \right] x^2 + \mathcal{O}(x^4).
$$
Hence $x = 0$ is the point of maximum of $u_{\rm cn-rogue}(x,0)$. Defining the magnification number as
$$
M_{\rm cn}(k) = \frac{u_{\rm cn-rogue}(0,0)}{\max\limits_{x \in [-2K(k),2K(k)]} |u_{\rm cn}(x,0)|} = 3,
$$
we obtain the expression in (\ref{magnification-factor}). The magnification factor is independent of the
amplitude of the {\em cn}-periodic wave.

\appendix

\section{Proof of $N$-fold Darboux transformation}
\label{appendix-a}

Here we prove Theorem \ref{theorem-N-fold} with explicit algebraic computations.
The Darboux transformation matrix $T(\lambda)$ in (\ref{Darboux}) is sought in the following explicit form:
\begin{equation}
\label{T-lambda}
T(\lambda) = I + \sum_{k=1}^N \frac{1}{\lambda - \lambda_k} T_k, \quad T_k = \widetilde{\varphi}^{(k)} \otimes (\varphi^{(k)})^t \sigma_1 \sigma_3,
\end{equation}
where the sign $\otimes$ denotes the outer vector product and $I$ denotes an identity $2 \times 2$ matrix.

We note that $\varphi^{(k)} \in {\rm ker}(T_k)$ and $\widetilde{\varphi}^{(k)} \in {\rm ran}(T_k)$.
It is assumed in Theorem \ref{theorem-N-fold} that $\varphi^{(k)} = (p_k,q_k)^t$, $1 \leq k \leq N$ is a particular
smooth nonzero solution to system (\ref{3.2}) and (\ref{3.3}) with fixed $\lambda = \lambda_k \in \mathbb{C} \backslash \{0\}$
satisfying $\lambda_k \neq \pm \lambda_j$ for every $k \neq j$, whereas
$\{\widetilde{\varphi}^{(k)}\}_{1 \leq k \leq N}$ is a unique
solution of the linear algebraic system (\ref{lin-system}). Deeper in the proof,
we will be able to show that $\widetilde{\varphi}^{(k)} = (\widetilde{p}_k,\widetilde{q}_k)^t$, $1 \leq k \leq N$
is a particular solution to system (\ref{3.2}) and (\ref{3.3}) with $\lambda = \lambda_k$ and new potential $\widetilde{u}$
given by the transformation formula (\ref{N-fold}).

First, let us show that the two lines in the definition (\ref{N-fold}) are identical. Let us define
entries of the matrix $A$ by (\ref{entry-A}). Each entry is finite, moreover, $A_{jk} = A_{kj}$.
The linear system (\ref{lin-system}) can be split into two parts as follows
\begin{equation}
\label{lin-system-A}
\sum_{k=1}^N \frac{\langle \varphi^{(j)}, \varphi^{(k)} \rangle}{\lambda_j + \lambda_k} \widetilde{p}_k = q_j, \quad
\sum_{k=1}^N \frac{\langle \varphi^{(j)}, \varphi^{(k)} \rangle}{\lambda_j + \lambda_k} \widetilde{q}_k = -p_j.
\end{equation}
Thanks to the symmetry of $A$, we obtain from (\ref{lin-system-A}):
\begin{equation}
\label{symm1}
\sum_{j=1}^N \widetilde{q}_j q_j =  \sum_{j=1}^N  \sum_{k=1}^N \frac{\langle \varphi^{(j)}, \varphi^{(k)} \rangle}{\lambda_j + \lambda_k}
\widetilde{p}_k \widetilde{q}_j = \sum_{j=1}^N  \sum_{k=1}^N \frac{\langle \varphi^{(j)}, \varphi^{(k)} \rangle}{\lambda_j + \lambda_k}
\widetilde{p}_j \widetilde{q}_k = -\sum_{j=1}^N \widetilde{p}_j p_j.
\end{equation}
This proves that the two lines in the definition (\ref{N-fold}) are identical. For further use, let us also
derive another relation from the system (\ref{lin-system-A}):
\begin{eqnarray}
\nonumber
\sum_{j=1}^N \lambda_j \widetilde{q}_j q_j - \sum_{j=1}^N \lambda_j \widetilde{p}_j p_j & = &
\sum_{j=1}^N  \sum_{k=1}^N \langle \varphi^{(j)}, \varphi^{(k)} \rangle \widetilde{p}_k \widetilde{q}_j \\
\nonumber & = & \left( \sum_{j=1}^N \widetilde{p}_j p_j \right)
\left( \sum_{k=1}^N p_k \widetilde{q}_k \right) + \left( \sum_{j=1}^N \widetilde{q}_j q_j \right)
\left( \sum_{k=1}^N \widetilde{p}_k q_k \right) \\
\label{symm2} & = & \frac{1}{2} (\widetilde{u}-u)
\left( \sum_{k=1}^N p_k \widetilde{q}_k  - \sum_{k=1}^N \widetilde{p}_k q_k \right).
\end{eqnarray}

Next, we show validity of the Darboux equation (\ref{darboux-1}) under the transformation formula (\ref{T-lambda}).
Substituting (\ref{T-lambda}) to (\ref{darboux-1}) yields the following equations at the simple poles
\begin{equation}
\label{T-k}
\partial_x T_k + T_k U(\lambda_k,u) = U(\lambda_k,\widetilde{u}) T_k, \quad 1 \leq k \leq N,
\end{equation}
and the following equation at the constant term
\begin{equation}
\label{T-0}
\widetilde{u} \sigma_3 \sigma_1 = u \sigma_3 \sigma_1  + \sum_{k=1}^N T_k \sigma_3 - \sigma_3 T_k.
\end{equation}
Equation (\ref{T-0}) yields (\ref{N-fold}) due to representation (\ref{T-lambda}).

Let us show that equations (\ref{T-k}) are satisfied if $\varphi^{(k)}$ solves (\ref{3.2}) with
$\lambda = \lambda_k$ and $u$, whereas $\widetilde{\varphi}^{(k)}$ solves (\ref{3.2}) with
$\lambda = \lambda_k$ and $\widetilde{u}$.
Recall that $\sigma_1 \sigma_3 = -\sigma_3 \sigma_1$ and $\sigma_1 \sigma_1 = \sigma_3 \sigma_3 = I$.
Substituting (\ref{T-lambda}) to both sides of (\ref{T-k}) yields
\begin{eqnarray*}
& \phantom{t} & \left[ \partial_x \widetilde{\varphi}^{(k)} \right] \otimes (\varphi^{(k)})^t \sigma_1 \sigma_3 +
\widetilde{\varphi}^{(k)} \otimes \left[ \partial_x (\varphi^{(k)})^t  \right] \sigma_1 \sigma_3
+ \widetilde{\varphi}^{(k)} \otimes (\varphi^{(k)})^t \sigma_1 \sigma_3 U(\lambda_1,u) \\
& \phantom{t} &
= \left[ \partial_x \widetilde{\varphi}^{(k)} \right] \otimes (\varphi^{(k)})^t \sigma_1 \sigma_3 +
\widetilde{\varphi}^{(k)} \otimes \left[ \partial_x (\varphi^{(k)})^t  \right] \sigma_1 \sigma_3
- \widetilde{\varphi}^{(k)} \otimes (\varphi^{(k)})^t U(\lambda_1,u)^t \sigma_1 \sigma_3 \\
& \phantom{t} & =
\left[ \partial_x \widetilde{\varphi}^{(k)} \right] \otimes (\varphi^{(k)})^t \sigma_1 \sigma_3
\end{eqnarray*}
and
$$
U(\lambda_1,\widetilde{u}) \widetilde{\varphi}^{(k)} \otimes (\varphi^{(k)})^t \sigma_1 \sigma_3 =
\left[ \partial_x \widetilde{\varphi}^{(k)} \right] \otimes (\varphi^{(k)})^t \sigma_1 \sigma_3,
$$
hence equation (\ref{T-k}) is satisfied.

We show now that if $\{ \varphi^{(k)} \}_{1 \leq k \leq N}$ solve (\ref{3.2}) with
$\{\lambda_k\}_{1 \leq k \leq N}$ and $u$ and $\{ \widetilde{\varphi}^{(k)} \}_{1 \leq k \leq N}$ are obtained from
the linear algebraic system (\ref{lin-system}), then
$\{ \widetilde{\varphi}^{(k)} \}_{1 \leq k \leq N}$ solve (\ref{3.2})  with
$\{\lambda_k\}_{1 \leq k \leq N}$  and $\widetilde{u}$. We note from the linear system (\ref{3.2})
that
\begin{equation}
\label{x-derivative-eq}
\partial_x \langle \varphi^{(j)}, \varphi^{(k)} \rangle = (\lambda_j + \lambda_k) \langle \varphi^{(j)}, \sigma_3 \varphi^{(k)} \rangle.
\end{equation}
Differentiating (\ref{lin-system}) in $x$ and substituting (\ref{operator-U}) and (\ref{x-derivative-eq}) yield
\begin{eqnarray}
\nonumber
& \phantom{t} & \sum_{k=1}^N \frac{\langle \varphi^{(j)}, \varphi^{(k)} \rangle}{\lambda_j + \lambda_k} \left[ \partial_x
\widetilde{\varphi}^{(k)} - \lambda_k \sigma_3 \widetilde{\varphi}^{(k)} - \widetilde{u} \sigma_3 \sigma_1 \widetilde{\varphi}^{(k)} \right]  \\
& = & (\widetilde{u}-u) \varphi^{(j)} - \sum_{k=1}^N \left[ \langle \varphi^{(j)}, \sigma_3 \varphi^{(k)} \rangle \widetilde{\varphi}^{(k)}
+ \langle \varphi^{(j)}, \varphi^{(k)} \rangle \sigma_3 \widetilde{\varphi}^{(k)} \right] = 0, \label{x-relation-cancel}
\end{eqnarray}
where the last equality is due to the transformation formula (\ref{N-fold}). Thus, if the linear
system (\ref{lin-system}) is assumed to admit a unique solution, then
$\widetilde{\varphi}^{(k)}$ solves (\ref{3.2})  with
$\lambda = \lambda_k$  and $\widetilde{u}$.

It remains to show validity of the Darboux equation (\ref{darboux-2}) under the transformation formula (\ref{T-lambda}).
Substituting (\ref{T-lambda}) to (\ref{darboux-2})
yields the following equations at the simple poles
\begin{equation}
\label{T-k-time}
\partial_t T_k + T_k V(\lambda_k,u) = V(\lambda_k,\widetilde{u}) T_k, \quad 1 \leq k \leq N,
\end{equation}
the same equation (\ref{T-0}) at $\lambda^2$ and the following two equations
at $\lambda^1$ and $\lambda^0$ respectively:
\begin{equation}
\label{T-1}
\widetilde{u}^2 \sigma_3 + \widetilde{u}_x \sigma_1 + 2 \widetilde{u} \sum_{k=1}^N \sigma_3 \sigma_1 T_k
= u^2 \sigma_3 + u_x \sigma_1 + 2 u \sum_{k=1}^N T_k \sigma_3 \sigma_1 + 2 \sum_{k=1}^N \lambda_k ( T_k \sigma_3 - \sigma_3 T_k )
\end{equation}
and
\begin{eqnarray}
\nonumber
 (2 \widetilde{u}^3 + \widetilde{u}_{xx}) \sigma_3 \sigma_1 + 2 \widetilde{u}^2 \sum_{k=1}^N \sigma_3 T_k
+ 2 \widetilde{u}_x \sum_{k=1}^N \sigma_1 T_k + 4 \widetilde{u} \sum_{k=1}^N \lambda_k \sigma_3 \sigma_1 T_k + 4 \sum_{k=1}^N \lambda_k^2
\sigma_3 T_k & \phantom{t} &\\
= (2 u^3 + u_{xx}) \sigma_3 \sigma_1 + 2 u^2 \sum_{k=1}^N T_k \sigma_3
+ 2 u_x \sum_{k=1}^N T_k \sigma_1 + 4 u \sum_{k=1}^N \lambda_k T_k \sigma_3 \sigma_1 +
4 \sum_{k=1}^N \lambda_k^2 T_k \sigma_3. & \phantom{t} &
\label{T-2}
\end{eqnarray}

Let us show that equations (\ref{T-k-time}) are satisfied if $\varphi^{(k)}$ solves (\ref{3.3}) with
$\lambda = \lambda_k$ and $u$, whereas $\widetilde{\varphi}^{(k)}$ solves (\ref{3.3}) with
$\lambda = \lambda_k$ and $\widetilde{u}$. Substituting (\ref{T-lambda}) to both sides of (\ref{T-k-time}) yields
\begin{eqnarray*}
& \phantom{t} & \left[ \partial_t \widetilde{\varphi}^{(k)}  \right] \otimes (\varphi^{(k)})^t \sigma_1 \sigma_3 +
\widetilde{\varphi}^{(k)} \otimes \left[ \partial_t (\varphi^{(k)})^t  \right] \sigma_1 \sigma_3
+ \widetilde{\varphi}^{(k)} \otimes (\varphi^{(k)})^t \sigma_1 \sigma_3 V(\lambda_1,u) \\
& \phantom{t} &
= \left[ \partial_t \widetilde{\varphi}^{(k)} \right] \otimes (\varphi^{(k)})^t \sigma_1 \sigma_3 +
\widetilde{\varphi}^{(k)} \otimes \left[ \partial_t (\varphi^{(k)})^t  \right] \sigma_1 \sigma_3
- \widetilde{\varphi}^{(k)} \otimes (\varphi^{(k)})^t V(\lambda_1,u)^t \sigma_1 \sigma_3 \\
& \phantom{t} & =
\left[ \partial_t \widetilde{\varphi}^{(k)} \right] \otimes (\varphi^{(k)})^t \sigma_1 \sigma_3
\end{eqnarray*}
and
$$
V(\lambda_1,\widetilde{u}) \widetilde{\varphi}^{(k)} \otimes (\varphi^{(k)})^t \sigma_1 \sigma_3 =
\left[ \partial_t \widetilde{\varphi}^{(k)} \right] \otimes (\varphi^{(k)})^t \sigma_1 \sigma_3,
$$
hence equation (\ref{T-k-time}) is satisfied.

In order to show the validity of equation (\ref{T-1}), we differentiate (\ref{T-0}) in $x$
and substitute (\ref{T-k}) to obtain
\begin{equation}
\label{T-25}
(\widetilde{u}_x-u_x)\sigma_1 = 2\sum\limits_{k=1}^N\lambda_k(T_k\sigma_3-\sigma_3T_k)
+ \widetilde{u}\sum\limits_{k=1}^N(\sigma_1T_k\sigma_3-\sigma_3\sigma_1T_k)
+ u \sum\limits_{k=1}^N(\sigma_3T_k\sigma_1+T_k\sigma_3\sigma_1).
\end{equation}
Substituting (\ref{T-25}) into (\ref{T-1}) yields a simplified form of the equation:
\begin{equation}
\label{T-101}
(\widetilde{u}^2 - u^2) \sigma_3 + \widetilde{u} \sum_{k=1}^N \sigma_1 T_k \sigma_3 + \sigma_3 \sigma_1 T_k
+ u \sum_{k=1}^N \sigma_3 T_k \sigma_1 - T_k \sigma_3 \sigma_1 = 0.
\end{equation}
Further substituting (\ref{T-0}) into (\ref{T-101}) yields
\begin{equation}
\label{T-100}
\sum\limits_{k=1}^N(\sigma_1T_k\sigma_3+\sigma_3\sigma_1T_k+T_k\sigma_3\sigma_1-\sigma_3T_k\sigma_1) = 0.
\end{equation}
The validity of equation (\ref{T-100}) is satisfied thanks again to equation (\ref{T-0}):
\begin{equation}
\label{TTT}
(\widetilde{u}-u) \sigma_3 = \sum_{k=1}^N (T_k \sigma_3 \sigma_1 - \sigma_3 T_k \sigma_1 ), \quad
(u-\widetilde{u}) \sigma_3 = \sum_{k=1}^N (\sigma_1 T_k \sigma_3 + \sigma_3 \sigma_1 T_k).
\end{equation}
Hence, equation (\ref{T-1}) is satisfied.

In order to show the validity of equation (\ref{T-2}), we differentiate (\ref{T-25}) in $x$
and substitute (\ref{T-k}) to obtain
\begin{eqnarray}
\nonumber
(\tilde{u}_{xx}-u_{xx})\sigma_3\sigma_1 & = & 4\sum\limits_{k=1}^N\lambda_k^2(T_k\sigma_3-\sigma_3T_k)+
2\tilde{u}\sum\limits_{k=1}^N\lambda_k(\sigma_1T_k\sigma_3-\sigma_3\sigma_1T_k) \\
\nonumber
&& +2u\sum\limits_{k=1}^N\lambda_k(\sigma_3T_k\sigma_1+T_k\sigma_3\sigma_1)+\tilde{u}_x\sum\limits_{k=1}^N(\sigma_3\sigma_1T_k\sigma_3-\sigma_1T_k)\\
\nonumber
&& +u_x\sum\limits_{k=1}^N(T_k\sigma_1+\sigma_3T_k\sigma_3\sigma_1)+(\tilde{u}^2+u^2)\sum\limits_{k=1}^N(\sigma_3T_k-T_k\sigma_3)\\
\label{T-21}
&& +2u\tilde{u}\sum\limits_{k=1}^N(\sigma_3\sigma_1T_k\sigma_1+\sigma_1T_k\sigma_3\sigma_1).
\end{eqnarray}
Substituting (\ref{T-0}) and (\ref{T-21}) into (\ref{T-2}) yields a simplified form of the equation:
\begin{eqnarray}
\nonumber
& \phantom{t} & 2\widetilde{u}\sum\limits_{k=1}^N\lambda_k(\sigma_1T_k\sigma_3+\sigma_3\sigma_1T_k)+
2u\sum\limits_{k=1}^N\lambda_k(\sigma_3T_k\sigma_1-T_k\sigma_3\sigma_1) \\
\nonumber
& \phantom{t} & + \widetilde{u}_x\sum\limits_{k=1}^N(\sigma_3\sigma_1T_k\sigma_3+\sigma_1T_k)
+u_x\sum\limits_{k=1}^N(\sigma_3T_k\sigma_3\sigma_1-T_k\sigma_1) +(\widetilde{u}^2-u^2) \sum\limits_{k=1}^N(T_k\sigma_3+\sigma_3T_k)\\
& \phantom{t} &
+2 \widetilde{u} u \sum\limits_{k=1}^N (\sigma_3 \sigma_1 T_k \sigma_1 + \sigma_1 T_k \sigma_3 \sigma_1 + T_k \sigma_3 - \sigma_3 T_k) = 0.
\label{TT-1}
\end{eqnarray}
The last term in the left-hand side of (\ref{TT-1}) is identically zero thanks to equation (\ref{T-100})
after multiplication by $\sigma_1$ on the right. Multiplication of equation (\ref{T-100}) by $\sigma_3$ on the right
allows us to group the terms containing $u_x$ and $\widetilde{u}_x$. As a result,
we rewrite (\ref{TT-1}) in the equivalent form
\begin{eqnarray}
\nonumber
& \phantom{t} &
2\widetilde{u}\sum\limits_{k=1}^N\lambda_k(\sigma_1T_k\sigma_3+\sigma_3\sigma_1T_k)+2u\sum\limits_{k=1}^N\lambda_k(\sigma_3T_k\sigma_1-T_k\sigma_3\sigma_1) \\
& \phantom{t} & + (u_x-\widetilde{u}_x) \sum\limits_{k=1}^N(\sigma_3T_k\sigma_3\sigma_1-T_k\sigma_1)
+(\widetilde{u}^2-u^2) \sum\limits_{k=1}^N(T_k\sigma_3+\sigma_3T_k) = 0. \label{TT-2}
\end{eqnarray}
Multiplying (\ref{T-25}) by $\sigma_1$ from the left and from the right, we obtain
\begin{eqnarray*}
(\widetilde{u}_x-u_x) I = 2\sum\limits_{k=1}^N\lambda_k ( \sigma_1 T_k\sigma_3 + \sigma_3 \sigma_1 T_k)
+ \widetilde{u}\sum\limits_{k=1}^N (T_k \sigma_3 + \sigma_3 T_k)
+ u \sum\limits_{k=1}^N (\sigma_1 \sigma_3T_k\sigma_1 + \sigma_1 T_k\sigma_3\sigma_1)
\end{eqnarray*}
and
\begin{eqnarray*}
(\widetilde{u}_x-u_x) I =  2\sum\limits_{k=1}^N\lambda_k ( T_k \sigma_3 \sigma_1 - \sigma_3 T_k \sigma_1 )
+ \widetilde{u}\sum\limits_{k=1}^N(\sigma_1 T_k \sigma_3 \sigma_1 -\sigma_3\sigma_1T_k \sigma_1 )
+ u \sum\limits_{k=1}^N(\sigma_3T_k+T_k\sigma_3),
\end{eqnarray*}
from which one can rewrite (\ref{TT-2}) in the equivalent form
$$
(\widetilde{u}_x - u_x) (\widetilde{u} - u) I - (\widetilde{u}_x - u_x) \sum\limits_{k=1}^N(\sigma_3T_k\sigma_3\sigma_1-T_k\sigma_1) = 0,
$$
which is satisfied thanks to equation (\ref{TTT}). Hence, equation (\ref{T-2}) is satisfied.

Finally, we show that if $\{ \varphi^{(k)} \}_{1 \leq k \leq N}$ solve (\ref{3.3}) with
$\{\lambda_k\}_{1 \leq k \leq N}$ and $u$ and $\{ \widetilde{\varphi}^{(k)} \}_{1 \leq k \leq N}$ are obtained from
the linear algebraic system (\ref{lin-system}), then
$\{ \widetilde{\varphi}^{(k)} \}_{1 \leq k \leq N}$ solve (\ref{3.3})  with
$\{\lambda_k\}_{1 \leq k \leq N}$  and $\widetilde{u}$. We note from the linear system (\ref{3.3})
that
\begin{eqnarray}
\nonumber
\partial_t \langle \varphi^{(j)}, \varphi^{(k)} \rangle & = & -(\lambda_j + \lambda_k)
\left[ 4 (\lambda_j^2 - \lambda_j \lambda_k + \lambda_k^2) + 2 u^2 \right]
\langle \varphi^{(j)}, \sigma_3 \varphi^{(k)} \rangle \\
\label{t-derivative-eq}
& \phantom{t} &
+ 4 (\lambda_j^2 - \lambda_k^2) u \langle \varphi^{(j)}, \sigma_3 \sigma_1 \varphi^{(k)} \rangle
- 2 (\lambda_j + \lambda_k) u_x \langle \varphi^{(j)}, \sigma_1 \varphi^{(k)} \rangle.
\end{eqnarray}
Differentiating (\ref{lin-system}) in $t$ and substituting (\ref{operator-V}) and (\ref{t-derivative-eq}) yield
{\small \begin{eqnarray}
\nonumber
& \phantom{t} & \sum_{k=1}^N \frac{\langle \varphi^{(j)}, \varphi^{(k)} \rangle}{\lambda_j + \lambda_k} \left[ \partial_t
\widetilde{\varphi}^{(k)} + (4 \lambda_k^3 + 2 \lambda_k \widetilde{u}^2) \sigma_3 \widetilde{\varphi}^{(k)}
+ 4 \lambda_k^2 \widetilde{u} \sigma_3 \sigma_1 \widetilde{\varphi}^{(k)}
+ 2 \lambda_k \widetilde{u}_x \sigma_1 \widetilde{\varphi}^{(k)}
+ (2 \widetilde{u}^3 + \widetilde{u}_{xx}) \sigma_3 \sigma_1 \widetilde{\varphi}^{(k)}\right]  \\
\nonumber
& = & \sum_{k=1}^N \left[ (4 \lambda_j^2 - 4 \lambda_j \lambda_k + 4 \lambda_k^2 + 2 u^2 )
\langle \varphi^{(j)}, \sigma_3 \varphi^{(k)} \rangle
+ 4 (\lambda_k - \lambda_j) u \langle \varphi^{(j)}, \sigma_3 \sigma_1 \varphi^{(k)} \rangle
+ 2 u_x \langle \varphi^{(j)}, \sigma_1 \varphi^{(k)} \rangle \right] \widetilde{\varphi}^{(k)}  \\
\nonumber
& + &  \sum_{k=1}^N \langle \varphi^{(j)}, \varphi^{(k)} \rangle \left[
(4 \lambda_k^2 - 4 \lambda_k \lambda_j + 4 \lambda_j^2 + 2 \widetilde{u}^2) \sigma_3 \widetilde{\varphi}^{(k)}
+ 4 (\lambda_k - \lambda_j) \widetilde{u} \sigma_3 \sigma_1 \widetilde{\varphi}^{(k)}
+ 2 \widetilde{u}_x \sigma_1 \widetilde{\varphi}^{(k)}\right] \\
& + &  2 \lambda_j (u^2 - \widetilde{u}^2) \sigma_1 \varphi^{(j)} + 4 \lambda_j^2 (u-\widetilde{u}) \varphi^{(j)}
- 2 \lambda_j (u_x - \widetilde{u}_x) \sigma_3 \varphi^{(j)}  + (2 u^3 + u_{xx} - 2 \widetilde{u}^3 - \widetilde{u}_{xx}) \varphi^{(j)}.
\label{t-equation}
\end{eqnarray}}
The terms proportional to $4 \lambda_j^2$ cancel out due to the same relation (\ref{x-relation-cancel}).
The terms proportional to $2 \lambda_j$ cancel out if the following relation is true:
\begin{eqnarray}
\nonumber
& \phantom{t} & (u^2 - \widetilde{u}^2) \sigma_1 \varphi^{(j)} + (\widetilde{u}_x - u_x) \sigma_3 \varphi^{(j)} =
2 \sum_{k=1}^N \lambda_k \left[ \langle \varphi^{(j)}, \sigma_3 \varphi^{(k)} \rangle \widetilde{\varphi}^{(k)}
+ \langle \varphi^{(j)}, \varphi^{(k)} \rangle \sigma_3 \widetilde{\varphi}^{(k)} \right] \\
\label{eq1-toshow}
& \phantom{t} &
+ 2 \widetilde{u} \sum_{k=1}^N \langle \varphi^{(j)}, \varphi^{(k)} \rangle
\sigma_3 \sigma_1 \widetilde{\varphi}^{(k)} + 2 u \sum_{k=1}^N \langle \varphi^{(j)}, \sigma_3 \sigma_1 \varphi^{(k)} \rangle \widetilde{\varphi}^{(k)}.
\end{eqnarray}
The other $\lambda_j$-independent terms cancel out if the following relation is true:
\begin{eqnarray}
\nonumber
& \phantom{t} & (2 \widetilde{u}^3 + \widetilde{u}_{xx} - 2 u^3 - u_{xx}) \varphi^{(j)}  =
4 \sum_{k=1}^N \lambda_k^2 \left[ \langle \varphi^{(j)}, \sigma_3 \varphi^{(k)} \rangle \widetilde{\varphi}^{(k)}
+ \langle \varphi^{(j)}, \varphi^{(k)} \rangle \sigma_3 \widetilde{\varphi}^{(k)} \right] \\
\nonumber & \phantom{t} &
+ 2 \widetilde{u}^2 \sum_{k=1}^N \langle \varphi^{(j)}, \varphi^{(k)} \rangle
\sigma_3 \widetilde{\varphi}^{(k)} + 2 u^2 \sum_{k=1}^N \langle \varphi^{(j)}, \sigma_3 \varphi^{(k)} \rangle \widetilde{\varphi}^{(k)} \\
\nonumber & \phantom{t} &
+ 4 \widetilde{u} \sum_{k=1}^N \lambda_k \langle \varphi^{(j)}, \varphi^{(k)} \rangle
\sigma_3 \sigma_1 \widetilde{\varphi}^{(k)} + 4 u \sum_{k=1}^N \lambda_k
\langle \varphi^{(j)}, \sigma_3 \sigma_1 \varphi^{(k)} \rangle \widetilde{\varphi}^{(k)} \\
& \phantom{t} &
+ 2 \widetilde{u}_x \sum_{k=1}^N \langle \varphi^{(j)}, \varphi^{(k)} \rangle
\sigma_1 \widetilde{\varphi}^{(k)} + 2 u_x \sum_{k=1}^N \langle \varphi^{(j)}, \sigma_1 \varphi^{(k)} \rangle \widetilde{\varphi}^{(k)}.
\label{eq2-toshow}
\end{eqnarray}
Provided equations (\ref{eq1-toshow}) and (\ref{eq2-toshow}) are satisfied,
the right-hand side of equation (\ref{t-equation}) is zero. If the linear
system (\ref{lin-system}) is assumed to admit a unique solution, then
$\widetilde{\varphi}^{(k)}$ solves (\ref{3.3})  with $\{\lambda_k\}_{1 \leq k \leq N}$  and $\widetilde{u}$.

Finally, we show validity of equations (\ref{eq1-toshow}) and (\ref{eq2-toshow}). In order to show (\ref{eq1-toshow}),
we first obtain the relation
\begin{equation}
\label{xx-derivative-eq}
\partial_x \langle \varphi^{(j)}, \sigma_3 \varphi^{(k)} \rangle =
(\lambda_j + \lambda_k) \langle \varphi^{(j)}, \varphi^{(k)} \rangle + 2u \langle \varphi^{(j)}, \sigma_1 \varphi^{(k)} \rangle,
\end{equation}
in addition to the relation (\ref{x-derivative-eq}). Then, we differentiate (\ref{x-relation-cancel}) in $x$,
substitute (\ref{operator-U}), (\ref{x-derivative-eq}), and (\ref{xx-derivative-eq}),
and obtain
\begin{eqnarray}
\nonumber
& \phantom{t} & (\widetilde{u}_x-u_x) \varphi^{(j)} + (\widetilde{u}-u) u \sigma_3 \sigma_1 \varphi^{(j)} =
2 \sum_{k=1}^N \lambda_k \left[ \langle \varphi^{(j)}, \sigma_3 \varphi^{(k)} \rangle \sigma_3 \widetilde{\varphi}^{(k)}
+ \langle \varphi^{(j)}, \varphi^{(k)} \rangle \widetilde{\varphi}^{(k)} \right] \\
& \phantom{t} &
+ \widetilde{u} \sum_{k=1}^N \left[ \langle \varphi^{(j)}, \sigma_3 \varphi^{(k)} \rangle \sigma_3 \sigma_1 \widetilde{\varphi}^{(k)}
+ \langle \varphi^{(j)}, \varphi^{(k)} \rangle \sigma_1 \widetilde{\varphi}^{(k)} \right] +
2 u \sum_{k=1}^N \langle \varphi^{(j)}, \sigma_1 \varphi^{(k)} \rangle \widetilde{\varphi}^{(k)},
\label{xx-relation-cancel}
\end{eqnarray}
where the relation (\ref{x-relation-cancel}) was used to cancel the $\lambda_j$ term.
By using the transformation formulas (\ref{N-fold}), we verify that
\begin{equation}
\label{xx-relation-useful}
(\widetilde{u} - u) \sigma_1 \varphi^{(j)}
=  \sum_{k=1}^N \left[ \langle \varphi^{(j)}, \sigma_1 \varphi^{(k)} \rangle \sigma_3 \widetilde{\varphi}^{(k)} -
 \langle \varphi^{(j)}, \sigma_3 \sigma_1 \varphi^{(k)} \rangle \widetilde{\varphi}^{(k)} \right].
\end{equation}
This allows us to simplify (\ref{xx-relation-cancel}) to the form
\begin{eqnarray}
\nonumber
(\widetilde{u}_x-u_x) \varphi^{(j)} & = &
2 \sum_{k=1}^N \lambda_k \left[ \langle \varphi^{(j)}, \sigma_3 \varphi^{(k)} \rangle \sigma_3 \widetilde{\varphi}^{(k)}
+ \langle \varphi^{(j)}, \varphi^{(k)} \rangle \widetilde{\varphi}^{(k)} \right] \\
\nonumber
& \phantom{t} & + \widetilde{u} \sum_{k=1}^N \left[ \langle \varphi^{(j)}, \sigma_3 \varphi^{(k)} \rangle \sigma_3 \sigma_1 \widetilde{\varphi}^{(k)}
+ \langle \varphi^{(j)}, \varphi^{(k)} \rangle \sigma_1 \widetilde{\varphi}^{(k)} \right] \\
& \phantom{t} & +
u \sum_{k=1}^N \left[ \langle \varphi^{(j)}, \sigma_1 \varphi^{(k)} \rangle \widetilde{\varphi}^{(k)}
+  \langle \varphi^{(j)}, \sigma_3 \sigma_1 \varphi^{(k)} \rangle \sigma_3 \widetilde{\varphi}^{(k)} \right].
\label{xxx-relation-cancel}
\end{eqnarray}
Substituting (\ref{xxx-relation-cancel})
to (\ref{eq1-toshow}) yields the following equation
\begin{eqnarray}
\nonumber
(u^2 - \widetilde{u}^2) \sigma_1 \varphi^{(j)} & = &
\widetilde{u} \sum_{k=1}^N \left[ \langle \varphi^{(j)}, \varphi^{(k)} \rangle \sigma_3 \sigma_1 \widetilde{\varphi}^{(k)}
-\langle \varphi^{(j)}, \sigma_3 \varphi^{(k)} \rangle \sigma_1 \widetilde{\varphi}^{(k)} \right] \\
& \phantom{t} &
+ u \sum_{k=1}^N \left[  \langle \varphi^{(j)}, \sigma_3 \sigma_1 \varphi^{(k)} \rangle \widetilde{\varphi}^{(k)}
-\langle \varphi^{(j)}, \sigma_1 \varphi^{(k)} \rangle \sigma_3 \widetilde{\varphi}^{(k)} \right].
 \label{eq11-toshow}
\end{eqnarray}
Thanks to the relations (\ref{x-relation-cancel}) and (\ref{xx-relation-useful}),
equation (\ref{eq11-toshow}) is satisfied, and so is equation (\ref{eq1-toshow}).

In order to show (\ref{eq2-toshow}), we first obtain the relations
\begin{equation}
\label{xx-derivative-eq1}
\partial_x \langle \varphi^{(j)}, \sigma_1 \varphi^{(k)} \rangle =
(\lambda_j - \lambda_k) \langle \varphi^{(j)}, \sigma_3 \sigma_1 \varphi^{(k)} \rangle - 2u \langle \varphi^{(j)}, \sigma_3 \varphi^{(k)} \rangle
\end{equation}
and
\begin{equation}
\label{xx-derivative-eq2}
\partial_x \langle \varphi^{(j)}, \sigma_3 \sigma_1 \varphi^{(k)} \rangle =
(\lambda_j - \lambda_k) \langle \varphi^{(j)}, \sigma_1 \varphi^{(k)} \rangle.
\end{equation}
Then, we differentiate (\ref{xxx-relation-cancel}) in $x$, substitute (\ref{operator-U}),
(\ref{x-derivative-eq}), (\ref{xx-derivative-eq}), (\ref{xx-derivative-eq1}), and (\ref{xx-derivative-eq2}),
and obtain
\begin{eqnarray}
\nonumber
& \phantom{t} & (\widetilde{u}_{xx} - u_{xx}) \varphi^{(j)} + u (\widetilde{u}_x - u_x) \sigma_3 \sigma_1 \varphi^{(j)} =
4 \sum_{k=1}^N \lambda_k^2 \left[ \langle \varphi^{(j)}, \sigma_3 \varphi^{(k)} \rangle \widetilde{\varphi}^{(k)}
+ \langle \varphi^{(j)}, \varphi^{(k)} \rangle \sigma_3  \widetilde{\varphi}^{(k)} \right] \\
\nonumber
& \phantom{t} & + 2 \widetilde{u} \sum_{k=1}^N \lambda_k
\left[ \langle \varphi^{(j)}, \sigma_3 \varphi^{(k)} \rangle \sigma_1 \widetilde{\varphi}^{(k)}
+ \langle \varphi^{(j)}, \varphi^{(k)} \rangle \sigma_3 \sigma_1 \widetilde{\varphi}^{(k)} \right]
+ 4 u \sum_{k=1}^N \lambda_k \langle \varphi^{(j)}, \sigma_1 \varphi^{(k)} \rangle \sigma_3 \widetilde{\varphi}^{(k)}  \\
\nonumber
& \phantom{t} & + \widetilde{u}_x \sum_{k=1}^N \left[ \langle \varphi^{(j)}, \sigma_3 \varphi^{(k)} \rangle \sigma_3 \sigma_1 \widetilde{\varphi}^{(k)}
+ \langle \varphi^{(j)}, \varphi^{(k)} \rangle \sigma_1 \widetilde{\varphi}^{(k)} \right] \\
\nonumber & \phantom{t} & +
u_x \sum_{k=1}^N \left[ \langle \varphi^{(j)}, \sigma_1 \varphi^{(k)} \rangle \widetilde{\varphi}^{(k)}
+  \langle \varphi^{(j)}, \sigma_3 \sigma_1 \varphi^{(k)} \rangle \sigma_3 \widetilde{\varphi}^{(k)} \right] \\
\nonumber & \phantom{t} &
+ \widetilde{u} u \sum_{k=1}^N \left[ 3 \langle \varphi^{(j)}, \sigma_1 \varphi^{(k)} \rangle \sigma_3 \sigma_1 \widetilde{\varphi}^{(k)}
+  \langle \varphi^{(j)}, \sigma_3 \sigma_1 \varphi^{(k)} \rangle \sigma_1 \widetilde{\varphi}^{(k)} \right] \\
& \phantom{t} &
- \widetilde{u}^2 \sum_{k=1}^N \left[ \langle \varphi^{(j)}, \sigma_3 \varphi^{(k)} \rangle \widetilde{\varphi}^{(k)}
+  \langle \varphi^{(j)}, \varphi^{(k)} \rangle \sigma_3 \widetilde{\varphi}^{(k)} \right]
- 2 u^2 \sum_{k=1}^N \langle \varphi^{(j)}, \sigma_3 \varphi^{(k)} \rangle \widetilde{\varphi}^{(k)},
\label{xxxx-relation-cancel}
\end{eqnarray}
where the relation (\ref{xxx-relation-cancel}) was used to cancel the $\lambda_j$ term. Substituting (\ref{xxxx-relation-cancel})
into (\ref{eq2-toshow}) and using (\ref{x-relation-cancel}) and (\ref{xx-relation-useful}) yield
\begin{eqnarray}
\nonumber
& \phantom{t} & 2 (\widetilde{u}^3 - u^3) \varphi^{(j)}  =
(2u - \widetilde{u}) (\widetilde{u}_x - u_x) \sigma_3 \sigma_1 \varphi^{(j)}
+ \widetilde{u}^2 \sum_{k=1}^N \left[ 3 \langle \varphi^{(j)}, \varphi^{(k)} \rangle \sigma_3 \widetilde{\varphi}^{(k)}
+  \langle \varphi^{(j)}, \sigma_3 \varphi^{(k)} \rangle \widetilde{\varphi}^{(k)}  \right] \\
\nonumber & \phantom{t} &
+ 4 u^2 \sum_{k=1}^N \langle \varphi^{(j)}, \sigma_3 \varphi^{(k)} \rangle \widetilde{\varphi}^{(k)}
- \widetilde{u} u \sum_{k=1}^N \left[ 3 \langle \varphi^{(j)}, \sigma_1 \varphi^{(k)} \rangle \sigma_3 \sigma_1 \widetilde{\varphi}^{(k)}
+  \langle \varphi^{(j)}, \sigma_3 \sigma_1 \varphi^{(k)} \rangle \sigma_1 \widetilde{\varphi}^{(k)} \right] \\
\nonumber & \phantom{t} &
+ 2 \widetilde{u} \sum_{k=1}^N \lambda_k \left[ \langle \varphi^{(j)}, \varphi^{(k)} \rangle
\sigma_3 \sigma_1 \widetilde{\varphi}^{(k)} - \langle \varphi^{(j)}, \sigma_3 \varphi^{(k)} \rangle \sigma_1 \widetilde{\varphi}^{(k)} \right]\\
& \phantom{t} &
+ 4 u \sum_{k=1}^N \lambda_k
\left[ \langle \varphi^{(j)}, \sigma_3 \sigma_1 \varphi^{(k)} \rangle \widetilde{\varphi}^{(k)}
- \langle \varphi^{(j)}, \sigma_1 \varphi^{(k)} \rangle \sigma_3 \widetilde{\varphi}^{(k)} \right].
\label{eq2-toshow-again}
\end{eqnarray}
Substituting (\ref{xxx-relation-cancel}) to (\ref{eq2-toshow-again}) yields
\begin{eqnarray}
\nonumber
& \phantom{t} & 2 (\widetilde{u} - u) (\widetilde{u}^2 + \widetilde{u} u + u^2) \varphi^{(j)}  =
2 \widetilde{u}^2 \sum_{k=1}^N \left[ \langle \varphi^{(j)}, \varphi^{(k)} \rangle \sigma_3 \widetilde{\varphi}^{(k)}
+  \langle \varphi^{(j)}, \sigma_3 \varphi^{(k)} \rangle \widetilde{\varphi}^{(k)}  \right] \\
\nonumber & \phantom{t} &
+ 2 \widetilde{u} u \sum_{k=1}^N \left[ \langle \varphi^{(j)}, \varphi^{(k)} \rangle \sigma_3 \widetilde{\varphi}^{(k)}
- \langle \varphi^{(j)}, \sigma_3 \varphi^{(k)} \rangle \widetilde{\varphi}^{(k)}
- 2 \langle \varphi^{(j)}, \sigma_1 \varphi^{(k)} \rangle \sigma_3 \sigma_1 \widetilde{\varphi}^{(k)}\right] \\
\nonumber & \phantom{t} &
+ 2 u^2 \sum_{k=1}^N \left[ 2 \langle \varphi^{(j)}, \sigma_3 \varphi^{(k)} \rangle \widetilde{\varphi}^{(k)}
+ \langle \varphi^{(j)}, \sigma_1 \varphi^{(k)} \rangle \sigma_3 \sigma_1 \widetilde{\varphi}^{(k)}
- \langle \varphi^{(j)}, \sigma_3 \sigma_1 \varphi^{(k)} \rangle \sigma_1 \widetilde{\varphi}^{(k)} \right] \\
\nonumber & \phantom{t} &
+ 4 u \sum_{k=1}^N \lambda_k
\left[ \langle \varphi^{(j)}, \sigma_3 \sigma_1 \varphi^{(k)} \rangle \widetilde{\varphi}^{(k)}
+ \langle \varphi^{(j)}, \varphi^{(k)} \rangle \sigma_3 \sigma_1  \widetilde{\varphi}^{(k)} \right. \\
& \phantom{t} & \phantom{texttexttext} \left.
- \langle \varphi^{(j)}, \sigma_1 \varphi^{(k)} \rangle \sigma_3 \widetilde{\varphi}^{(k)}
- \langle \varphi^{(j)}, \sigma_3 \varphi^{(k)} \rangle \sigma_1 \widetilde{\varphi}^{(k)} \right].
\label{eq2-toshow-again-again}
\end{eqnarray}
By using the relations (\ref{symm1}) and explicit computations, we obtain
\begin{eqnarray}
\nonumber & \phantom{t} &
 \sum_{k=1}^N \left[ \langle \varphi^{(j)}, \varphi^{(k)} \rangle \sigma_3 \widetilde{\varphi}^{(k)}
- \langle \varphi^{(j)}, \sigma_3 \varphi^{(k)} \rangle \widetilde{\varphi}^{(k)}
- 2 \langle \varphi^{(j)}, \sigma_1 \varphi^{(k)} \rangle \sigma_3 \sigma_1 \widetilde{\varphi}^{(k)}\right] \\
\label{last-1} & = & (\widetilde{u}-u) \varphi^{(j)} - 2 \left( \sum_{k=1}^N p_k \widetilde{q}_k - \sum_{k=1}^N \widetilde{p}_k q_k \right) \sigma_1 \varphi^{(j)},
\end{eqnarray}
\begin{eqnarray}
\nonumber & \phantom{t} &
\sum_{k=1}^N \left[ 2 \langle \varphi^{(j)}, \sigma_3 \varphi^{(k)} \rangle \widetilde{\varphi}^{(k)}
+ \langle \varphi^{(j)}, \sigma_1 \varphi^{(k)} \rangle \sigma_3 \sigma_1 \widetilde{\varphi}^{(k)}
- \langle \varphi^{(j)}, \sigma_3 \sigma_1 \varphi^{(k)} \rangle \sigma_1 \widetilde{\varphi}^{(k)} \right] \\
\label{last-2} & = & (\widetilde{u}-u) \varphi^{(j)} + 2 \left( \sum_{k=1}^N p_k \widetilde{q}_k - \sum_{k=1}^N \widetilde{p}_k q_k \right) \sigma_1 \varphi^{(j)}
\end{eqnarray}
and
\begin{eqnarray}
\nonumber & \phantom{t} &
\sum_{k=1}^N \lambda_k
\left[ \langle \varphi^{(j)}, \sigma_3 \sigma_1 \varphi^{(k)} \rangle \widetilde{\varphi}^{(k)}
+ \langle \varphi^{(j)}, \varphi^{(k)} \rangle \sigma_3 \sigma_1  \widetilde{\varphi}^{(k)} \right. \\
\nonumber & \phantom{t} & \left. \phantom{texttext}
- \langle \varphi^{(j)}, \sigma_1 \varphi^{(k)} \rangle \sigma_3 \widetilde{\varphi}^{(k)}
- \langle \varphi^{(j)}, \sigma_3 \varphi^{(k)} \rangle \sigma_1 \widetilde{\varphi}^{(k)} \right]\\
\label{last-3} & = &  - 2 \left( \sum_{k=1}^N \lambda_k \widetilde{p}_k p_k  - \sum_{k=1}^N \lambda_k
\widetilde{q}_k q_k \right) \sigma_1 \varphi^{(j)}.
\end{eqnarray}
Substituting (\ref{last-1}), (\ref{last-2}), and (\ref{last-3}) to (\ref{eq2-toshow-again-again})
cancel all terms thanks to the relations (\ref{symm2}) and (\ref{x-relation-cancel}).
Therefore, equation (\ref{eq2-toshow-again-again}) is satisfied, and so is equation (\ref{eq2-toshow}).

\end{document}